\documentclass[submission,copyright,creativecommons]{eptcs}
\providecommand{\event}{QPL 2025} 

\usepackage{mathdots}
\usepackage{circuitikz}
\usepackage{amsmath}
\usepackage{amssymb}
\usepackage{amsfonts}
\usepackage{stmaryrd}
\usepackage{mathtools}

\usetikzlibrary{arrows,shapes.gates.logic.US,shapes.gates.logic.IEC,calc}
\usetikzlibrary{shapes.geometric}
\usetikzlibrary{patterns}
\usetikzlibrary{fit}
\usetikzlibrary{positioning}
\usetikzlibrary{calc}
\usetikzlibrary{arrows}
\usetikzlibrary{decorations.markings}
\usetikzlibrary{decorations.pathreplacing}
\usetikzlibrary{shapes}
\usetikzlibrary{shapes.misc}
\usetikzlibrary{positioning}

\newif\ifpgfshaperectangleroundnortheast
\newif\ifpgfshaperectangleroundnorthwest
\newif\ifpgfshaperectangleroundsoutheast
\newif\ifpgfshaperectangleroundsouthwest
\pgfkeys{/pgf/.cd,
  rectangle round north east/.is if=pgfshaperectangleroundnortheast,
  rectangle round north west/.is if=pgfshaperectangleroundnorthwest,
  rectangle round south east/.is if=pgfshaperectangleroundsoutheast,
  rectangle round south west/.is if=pgfshaperectangleroundsouthwest,
  rectangle round north east, rectangle round north west,
  rectangle round south east, rectangle round south west,
}
\makeatletter
\def\pgf@sh@bg@rectangle{%
  \pgfkeysgetvalue{/pgf/outer xsep}{\outerxsep}%
  \pgfkeysgetvalue{/pgf/outer ysep}{\outerysep}%
  \pgfpathmoveto{\pgfpointadd{\southwest}{\pgfpoint{\outerxsep}{\outerysep}}}%
  {\ifpgfshaperectangleroundnorthwest\else\pgfsetcornersarced{\pgfpointorigin}\fi%
    \pgfpathlineto{\pgfpointadd{\southwest\pgf@xa=\pgf@x\northeast\pgf@x=\pgf@xa}{\pgfpoint{\outerxsep}{-\outerysep}}}}%
  {\ifpgfshaperectangleroundnortheast\else\pgfsetcornersarced{\pgfpointorigin}\fi%
    \pgfpathlineto{\pgfpointadd{\northeast}{\pgfpoint{-\outerxsep}{-\outerysep}}}}%
  {\ifpgfshaperectangleroundsoutheast\else\pgfsetcornersarced{\pgfpointorigin}\fi%
    \pgfpathlineto{\pgfpointadd{\southwest\pgf@ya=\pgf@y\northeast\pgf@y=\pgf@ya}{\pgfpoint{-\outerxsep}{\outerysep}}}}%
  {\ifpgfshaperectangleroundsouthwest\else\pgfsetcornersarced{\pgfpointorigin}\fi%
    \pgfpathclose}}

\newcommand\gatename[1]{\ensuremath{\mathrm{#1}}}

\tikzstyle{strings}=[baseline={([yshift=-.5ex]current bounding box.center)}]
\tikzset{every picture/.style={scale=.75}, transform shape,strings}

\usepackage{amsmath}
\usepackage{graphicx}
\usepackage{tikz-cd}
\usepackage{amsfonts}
\usepackage{amsmath}
\usepackage{makecell}
\usepackage{longtable}
\usepackage{multicol}
\usepackage{makecell}
\usepackage{amsthm}
\usepackage{subcaption}
\usepackage{cleveref}
\usepackage{tabularx}
\usepackage{mwe}
\usepackage{hyperref}
\usepackage{braket}
\usepackage{stmaryrd}
\usepackage{float}
\usepackage{stackrel}
\usepackage{tikzit}

\definecolor{zxred}{RGB}{232, 165, 165}
\definecolor{zxgreen}{RGB}{216, 248, 216}
\definecolor{zxdarkgreen}{RGB}{90, 175, 90}
\definecolor{zxhad}{RGB}{255, 255, 130}
\definecolor{zxpaulibox}{RGB}{221, 221, 255}

\tikzstyle{dots-lr}=[font={\normalsize}]
\tikzstyle{none}=[font={\Large\boldmath}]
\tikzstyle{red-dot}=[fill=none, draw={rgb,255: red,255; green,0; blue,4}, shape=circle, scale=2.2]
\tikzstyle{Z}=[draw, fill=white, circle, scale=1, inner sep=0pt, minimum size=10pt, minimum size=1.2em, rounded corners=0.5em, inner sep=0.2em, outer sep=-0.2em]
\tikzstyle{black-dot}=[draw, fill=black, circle, scale=1.0]
\tikzstyle{grey-dot}=[draw={rgb,255: red,198; green,198; blue,198}, fill={rgb,255: red,198; green,198; blue,198}, circle, scale=1.0]
\tikzstyle{X}=[draw, fill={rgb:black,1;white,3}, text=black, circle, scale=1, inner sep=0pt, minimum size=10pt, tikzit fill={rgb,255: red,191; green,191; blue,191}, minimum size=1.2em, rounded corners=0.5em, inner sep=0.2em, outer sep=-0.2em]
\tikzstyle{scalar}=[draw, rounded corners=1ex, rectangle round south west=false, rectangle round south east=false, tikzit fill={rgb,255: red,129; green,253; blue,255}, tikzit shape=rectangle]
\tikzstyle{scalarop}=[draw, rounded corners=1ex, rectangle round north west=false, rectangle round north east=false, tikzit fill={rgb,255: red,116; green,172; blue,255}, tikzit shape=rectangle]
\tikzstyle{map}=[draw, color=black, fill=white, rectangle]
\tikzstyle{Z phase dot}=[draw, minimum size=5mm, inner sep=0mm, scale=1, fill=zxgreen, shape=circle, font={\Large\boldmath}, minimum size=1.2em, rounded corners=0.5em, inner sep=0.2em, outer sep=-0.2em]
\tikzstyle{X phase dot}=[draw, minimum size=5mm, inner sep=0mm, scale=1, fill=zxred, shape=circle, font={\Large\boldmath}, minimum size=1.2em, rounded corners=0.5em, inner sep=0.2em, outer sep=-0.2em]
\tikzstyle{X dot}=[draw, minimum size=5mm, inner sep=0mm, scale=1, fill=zxred, shape=circle, font={\Large\boldmath}, minimum size=1.2em, rounded corners=0.5em, inner sep=0.2em, outer sep=-0.2em]

\tikzstyle{Z dot}=[draw, minimum size=5mm, inner sep=0mm, scale=1, fill=zxgreen, shape=circle, font={\Large\boldmath}, minimum size=1.2em, rounded corners=0.5em, inner sep=0.2em, outer sep=-0.2em]
\tikzstyle{box edge}=[-, dashed, dash pattern=on 2pt off 0.5pt, thick, draw={rgb,255: red,203; green,192; blue,225}]
\tikzstyle{dual}=[-, draw={rgb,255: red,198; green,198; blue,198}, line width=0.35mm]
\tikzstyle{primal}=[-, line width=0.35mm]
\tikzstyle{red-line}=[-, draw={rgb,255: red,255; green,0; blue,4}, line width=0.5mm]
\tikzstyle{over}=[-, draw=white, line width=2.3mm]
\tikzstyle{dual-dot}=[-Circle, draw={rgb,255: red,198; green,198; blue,198}, line width=0.5mm]
\tikzstyle{primal-dot}=[-Circle, draw=black, line width=0.5mm]
\tikzstyle{green-line}=[-, dashed, draw=zxdarkgreen, line width=0.45mm]
\tikzstyle{red-line-}=[-, dashed, draw=zxred, line width=0.45mm]
\usepackage[all,cmtip]{xy}
\usepackage{tikz}
\newcommand{\catname}[1]{\mathbf{#1}}
\theoremstyle{definition}
\newtheorem{definition}{Definition}[section]
\newtheorem{example}{Example}
\numberwithin{example}{section}
\newtheorem{theorem}{Theorem}
\numberwithin{theorem}{section}
\newtheorem{corollary}{Corollary}
\numberwithin{corollary}{section}
\newtheorem{lemma}{Lemma}
\numberwithin{lemma}{section}

\numberwithin{remark}{section}

\newtheorem{proposition}{Proposition}

\newcommand\restr[2]{{
  \left.\kern-\nulldelimiterspace 
  #1 
  \vphantom{\big|} 
  \right|_{#2} 
  }}

\numberwithin{proposition}{section}

\usepackage{iftex}

\ifpdf
  \usepackage{underscore}         
  \usepackage[T1]{fontenc}        
\else
  \usepackage{breakurl}           
\fi

\title{String Diagrams for Defect-Based Surface Code Computing}
\author{
Mateusz Kupper%
\thanks{Corresponding author; author order is reverse-alphabetical, with the purpose of sowing confusion.}
\institute{Dept.~Informatics\\ University of Sussex}
\email{m.kupper}
\vspace*{-1ex}
\email{@sussex.ac.uk}
\and
Dominic Horsman
\institute{Dept.~Computer Science
            \\ University of Oxford}
\email{dom.horsman}
\vspace*{-1ex}
\email{@gmail.com}
\and
Chris Heunen
\institute{\hspace*{-1cm}School of Informatics\hspace*{-1cm}
            \\ \hspace*{-1cm}University of Edinburgh\hspace*{-1cm}}
\email{chris.heunen}
\vspace*{-1ex}
\email{@ed.ac.uk}
\and
Niel de Beaudrap
\institute{Dept.~Informatics\\ University of Sussex}
\email{jrd27}
\vspace*{-1ex}
\email{@sussex.ac.uk}
}
\def\titlerunning{String Diagrams for Defect-Based Surface Code Computing}
\def\authorrunning{M.\;Kupper, D.\;Horsman, C.\;Heunen, \& N.\;de\;Beaudrap}

\begin{document}
\maketitle
\begin{abstract}
    Surface codes are a popular choice for implementing fault-tolerant quantum computing.
    Two-qubit gates may be realised in these codes using only nearest-neighbour interactions, either by lattice surgery~\cite{Horsman-2012} or by braiding defects around each other~\cite{raussendorf2007topological}.
    The effect of lattice surgery operations may be simply  described~\cite{deBeaudrap2020zxcalculusis} using the ZX-calculus: a graphical language that has proven effective for program design and optimisation (see e.g.~\cite{Gidney2019efficientmagicstate}).
    In this work, we formalise a similar description via the ZX-calculus of defect braiding, as it is conventionally described.
    We define a graphical calculus $\catname{KNOT}$, denoting the logical effects (in the absence of byproduct operations) of defect braiding in surface codes: we show how these effects may be described via a fragment of ZX-calculus which we call the ($0$, $\pi$)-fragment.
    We then use a `doubling' construction to define a subtheory of $\catname{KNOT}$, more specialised to standard encoding techniques in the defect braiding literature. Within this subtheory, we encompass standard braiding techniques by families of ‘ribbon-like’ and ‘tangle-like’ diagrams, each with semantics distinct from $\catname{KNOT}$, in terms of the ($0$, $\pi$)-fragment of ZX diagrams (again in the absence of byproducts). These subtheories may be used interoperably, and are each sound and complete for the ($0$, $\pi$)-fragment of ZX diagrams. This provides a starting point to use the formal diagrammatics to analyse the operational effects of defect braiding procedures.
\end{abstract}

\vspace{-0.4cm}
\section{Introduction}

Topological quantum error correcting codes store quantum information in global properties of a system rather than in local degrees of freedom. The logical operators and the encoded space are determined by the topological features of an underlying lattice with boundary conditions. Surface codes are among the most widely studied quantum error correction codes due to their high threshold (tolerance to hardware error) and compatibility with multiple physical quantum computing paradigms \cite{raussendorf2006fault}.

Surface codes can be either planar, where individual `patches' of physical qubits stabilised by the nearest-neighbour surface code operations support a single logical qubit each \cite{Dennis-2002}, or defect-based \cite{raussendorf2006fault}. In the latter, logical qubits are supported by (usually) pairs of holes in the lattice of physical qubits, with multiple logical qubits sharing the same lattice `patch'. Two-qubit entangling operations are performed using nearest-neighbour-only operations either by lattice surgery \cite{Horsman-2012} for planar codes, or defect braiding for defect-based codes \cite{raussendorf2007topological, Fowler2009}. Lattice surgery has lower resource requirements in general than defect-based surface codes (e.g. \cite{Gidney2019efficientmagicstate}). However, braiding remains an important potential tool for future surface code architectures and for the development of novel procedures using the code.

Logical information in the surface code is defined by logical operators: chains or rings of physical qubits whose joint state is the state of the logical operator. Extended surfaces of logical operators over time are known as `correlation surfaces' (see e.g. \cite[\S3]{austingoyal2009}). In the planar code, logical operators are strings between two opposite boundaries of the code patch, with the type of operator (X or Z) dependent on the type of boundary (`rough' or `smooth'). In the defect-based code, logical operators are supported between the boundaries created in the lattice by the pair of `holes' (qubits removed from the stabilisation procedure). These also come in two flavours, `primal' and `dual', depending on the qubits removed.

Defect-based codes can be implemented as either 2D or 3D codes, homologically equivalent~\cite{raussendorf2007topological}. In 2D codes, physical qubits persist and are operated on over time; in 3D, they are prepared in an entangled state and then layers are measured out to create `clock ticks' of the code. Intuitively, one can think of the 3D code as simply the 2D code represented in spacetime. Another homological equivalence~\cite{raussendorf2006fault} is between double and single defect codes. In single defects, the role of the second defect is taken by the boundary of the entire surface (in practice limiting to two logical qubits per surface). Figure~\ref{fig:surface-codes} shows how the logical operators for the planar, single defect, and double defect codes are supported across their respective surfaces. Planar and defect-based qubits can even be interconverted, by judiciously increasing or measuring out the stabilised surface \cite[\S5]{Horsman-2012}.

Two-qubit operations perform interactions of logical operators and correlation surfaces. In lattice surgery, planar code \gatename{X} or \gatename{Z} operators are merged or split by straightforward operations that cut or combine operators from different patches. Braiding, however, interacts logical operators by moving defects around each other (either moving single defects around each other, or one of a double-defect pair moving through the space between the pair of another logical qubit). This causes the correlation surfaces to interact, performing a unitary two-qubit entangling gate (usually \gatename{CNOT}) provided the qubits are supported by pairs of opposite-type defects (primal and dual).

\begin{figure}[!t]
	\centering
	\begin{subfigure}[b]{0.4\textwidth}
	\includegraphics[width=1.0\textwidth]{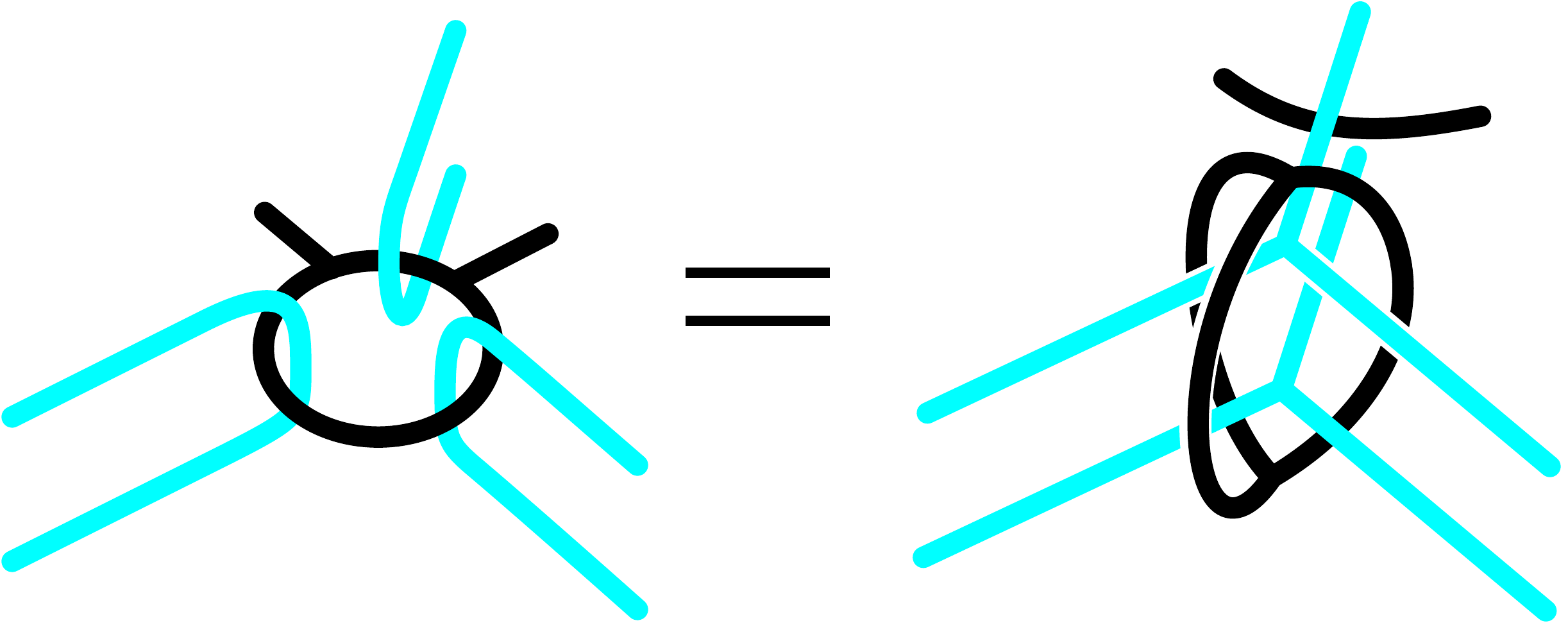}
	\caption{}
	\label{fig:surface3D}
	\end{subfigure}
	\hspace{0.03\textwidth}
	\begin{subfigure}[b]{0.4\textwidth}
	\includegraphics[width=1.0\textwidth]{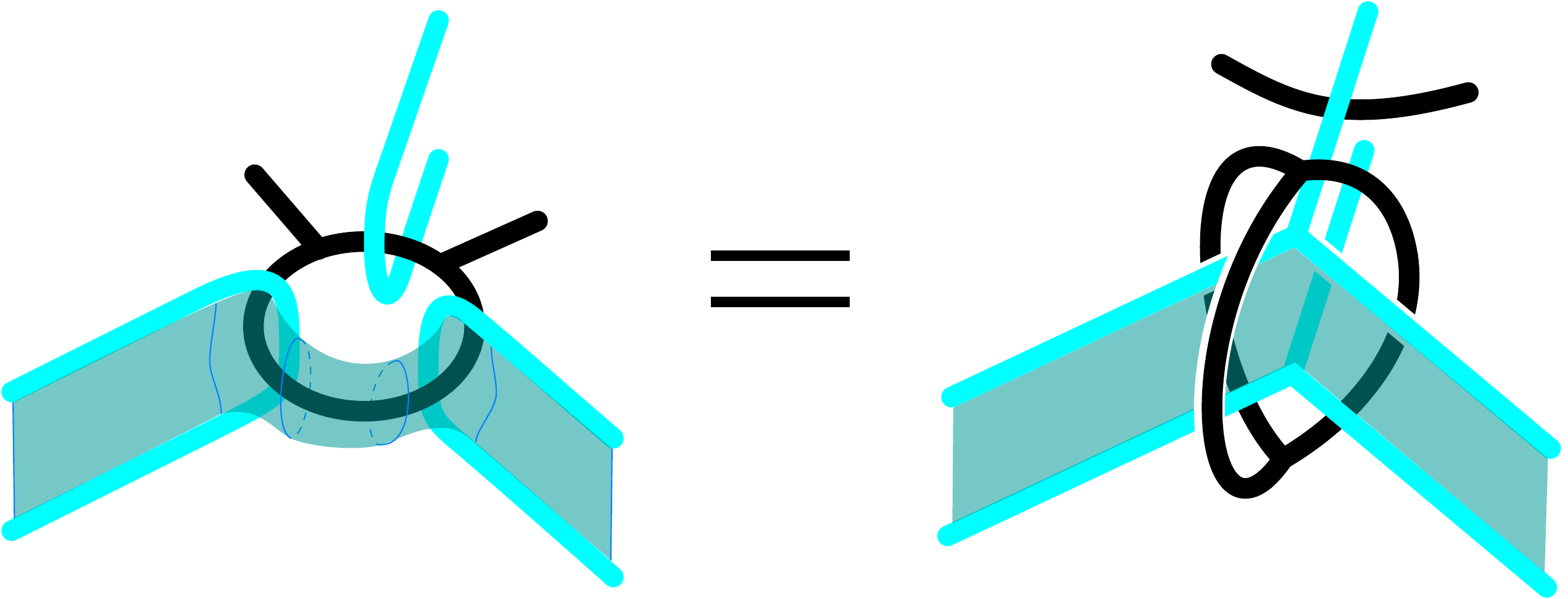}
	\caption{}
	\end{subfigure}
	\caption{Informal diagrammatic representations of braiding procedures using double defects. Different colours denote primal and dual; (b) shows also correlation surfaces. From Ref.~\cite{raussendorf2007topological}.}
	\label{fig:correlation-example-raussendorf}
\end{figure}

\begin{figure}[!b]
	\centering
	\includegraphics[width=0.7\textwidth]{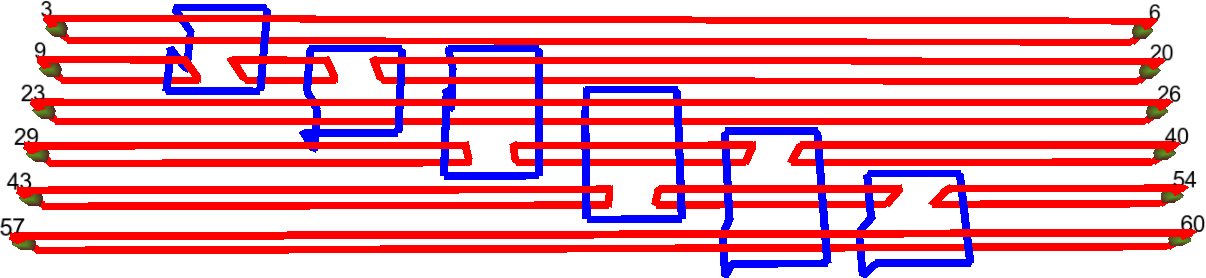}
	\caption{A three-dimensional representation of defect braiding. The red and blue structures represent dual and primal defects whose braiding implements a $T$ gate. From Ref.~\cite{paler2017fault}.}
	\label{fig:3d-braid}
\end{figure}

Braiding is clearly a more challenging operation to visualise than lattice surgery, not least being inherently 3D. Diagrams such as Figure~\ref{fig:3d-braid} show a simple subroutine of six \gatename{CNOT} operations (similar diagrams also appear in ~\cite{austingoyal2009, Austin2D, fujii2015quantum}). Larger procedures quickly become impractical to represent visually, both in terms of human readability and computing power needed to render 3D diagrams. Such diagrams are informal: they are not supported by a standardised semantics and compositional formal language capable of equational reasoning, but are instead validated by checking correlation surfaces ~\cite{raussendorf2007topological} or using a set of ad-hoc rewrite rules \cite{paler2019surfbraid}. This is made more difficult by the use of a combination of 2D and 3D diagrams and the use of informal syntactic sugar~\cite{raussendorf2007topological} (also see~Figure~\ref{fig:correlation-example-raussendorf}).

By contrast, lattice surgery benefits from a representation of its transformations in terms of ZX~diagrams~\cite{deBeaudrap2020zxcalculusis}.
This representation provides a formal, systematic means of reasoning about lattice surgery procedures~\cite{Gidney2020Video,hanks2020effective}. Diagrammatic reasoning (especially ZX-calculus) has been used in quantum error correction for the analysis of more general stabiliser codes (e.g. \cite{rodatz2024floquetifying, kissinger2024scalable, TownsendTeague2023,bombin2023unifying}) and surface codes \cite{kissinger2022phasefree}. There has been some limited work applying it to braided surface codes, where the correctness of \textit{some} braiding patterns has been proven using ZX-calculus \cite{Horsman-2011}. An informal mapping between different braiding patterns and ZX-calculus has also appeared in Ref.~\cite{hanks2020effective}. A further informal connection between lattice surgery, defect surgery, and ZX-calculus was made in Ref.~\cite{Gidney2020Video}. A non-diagrammatic method for equivalence checks for braiding patterns using Boolean expressions appeared in Ref.~\cite{paler2014design}.

This paper introduces a graphical language $\catname{KNOT}$, to formalize defect-based surface code computation. We build upon the informal diagrams introduced originally by Raussendorf~\cite{raussendorf2007topological} and presented in numerous forms in Refs.~\cite{fujii2015quantum, Austin2D, austingoyal2009} to define a formal language for defect braiding. We gather the rules appearing in these works and present them in a systematic way as a category with formal equivalences, and with our assumptions clearly stated. We also prove that the usual semantics appearing in the literature are sound (functorial) for ZX-calculus. In other words, by formalizing the largely intuitive and geometric arguments used in the literature, we provide a systematic account of defect braiding in the absence of byproduct operations. Furthermore, if we consider a pair of wires as encoding a qubit, we recover alternative semantics \cite{raussendorf2007topological} which is sound and complete for the ($0$, $\pi$)-fragment of ZX-calculus.

This formal language for braiding patterns will facilitate automated verification and optimisation of defect-based protocols, contributing to the simplifying analysis of topological quantum computation.

\vspace{-0.4cm}
\paragraph{Structure of the Paper.}

We proceed as follows. Section \ref{sec:scbraids} gives brief background to braided surface codes. Section~\ref{sec:braiding} introduces the the new formal framework for representing braiding operations using $\catname{KNOT}$. Section \ref{sec:interpretation} establishes different mappings into the ($0$, $\pi$)-fragment of the ZX calculus: from $\catname{KNOT}$, and from two subtheories of a `doubled' $\catname{KNOT}$, which lets us demonstrate the equivalence between these two subtheories of $\catname{KNOT}$ and the ($0$, $\pi$)-fragment of ZX calculus.

\vspace{-0.4cm}

\section{Surface codes and defect braiding}
\label{sec:scbraids}

\begin{figure}[t]
	\centering
	\begin{subfigure}[b]{0.3\textwidth}
	\includegraphics[width=\textwidth]{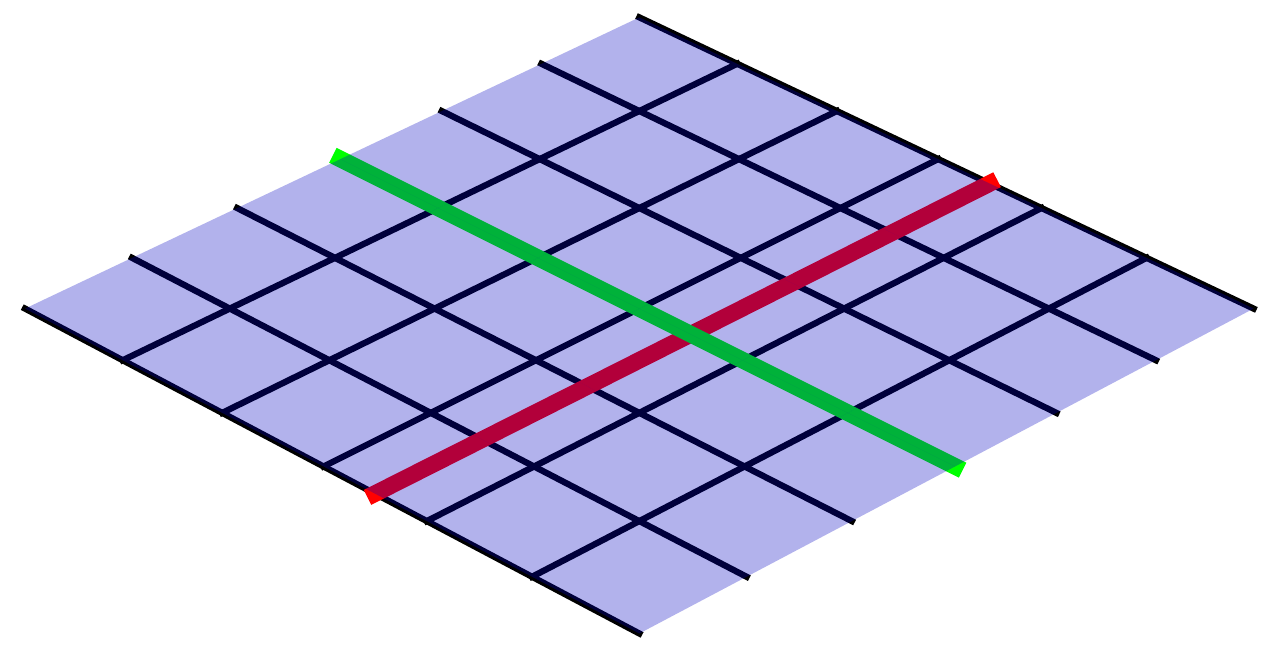}
	\caption{Surface code patch with two boundary types (rough - primal and smooth - dual) and two logical operators represented by red and green lines.}
	\label{fig:surface}
	\end{subfigure}
	\hspace{0.03\textwidth}
	\begin{subfigure}[b]{0.30\textwidth}
	\includegraphics[width=\textwidth]{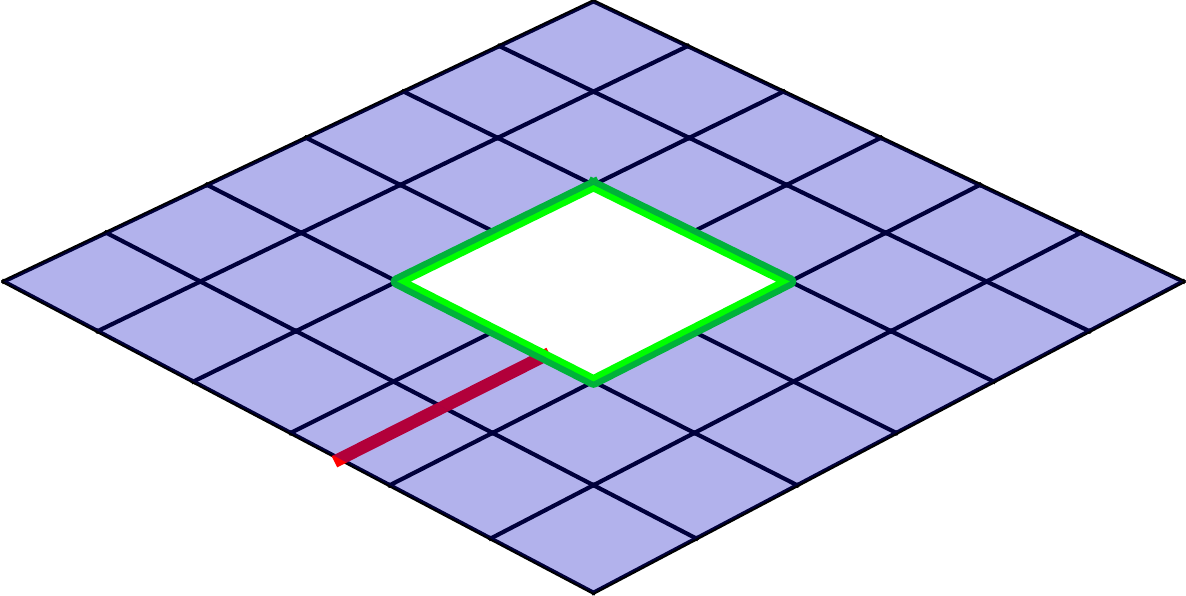}
	\caption{Surface code with one boundary type with a defect with the same boundary type. One of the logical operators encircles the defect.}
	\label{fig:surface-defect}
	\end{subfigure}
	\hspace{0.03\textwidth}
	\begin{subfigure}[b]{0.30\textwidth}
	\includegraphics[width=\textwidth]{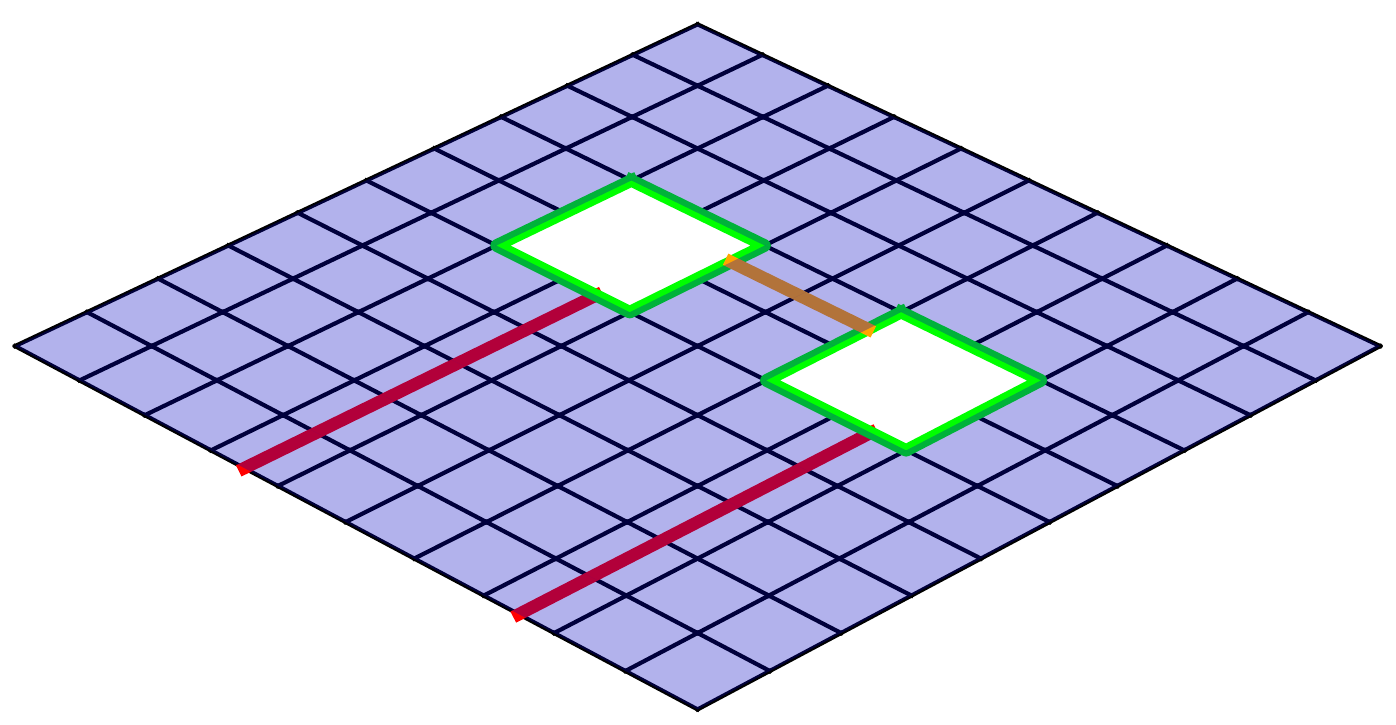}
	\caption{Surface code with one boundary type with two defects with the same boundary types. The red logical operators are together equivalent to the orange operator.}
	\label{fig:braiding}
	\end{subfigure}
	\caption{We can introduce defects and modify boundaries in a surface code to change the observables encoded in the code. The observables (or also called logical operators or errors) are supported on the boundaries and defects. The lattices and logical operators in each figure encode a single logical qubit.}
	\label{fig:surface-codes}
\end{figure}

A surface code lives on a square lattice whose edges carry physical qubits, see Figure~\ref{fig:surface-codes} \cite{bravyi1998quantumcodeslatticeboundary,Dennis-2002,raussendorf2006fault, Fowler2009,Austin2D}. The code space is the \(+1\)-eigenstate subspace of all stabiliser operators. Two families of commuting stabilisers act locally: vertex (\emph{star}) operators
\(A_v=\prod_{e\in\mathrm{star}(v)}\!X_e\)
and plaquette operators
\(B_p=\prod_{e\in\partial p}\!Z_e\).
Commutation ensures a consistent code space, with
\begin{equation*}
[A_v, A_{v'}] = [B_p, B_{p'}] = [A_v, B_p] = 0 \text{ for all } v, v', p, p'.
\end{equation*}
The code space is thus the subspace of the Hilbert space \(\mathcal{H}\) where all stabilisers satisfy
\begin{equation*}
A_v |\psi\rangle = |\psi\rangle, \quad B_p |\psi\rangle = |\psi\rangle  \text{ for all } v, p.
\end{equation*}
Error detection is given by changes of these stabiliser measurements over time.

\textbf{Defects} are regions of the lattice where stabilisers are deliberately not measured, creating internal boundaries. This creates spare degrees of freedom in the lattice which are used to encode logical qubits. Smooth defects (or dual defects) correspond to missing \(A_v\) stabilisers, while rough defects (or primal defects) correspond to missing \(B_p\) stabilisers. Stabiliser patterns may change from a one measurement round to the next, so defects can appear, disappear, merge, divide, or move.

\textbf{Logical operators} are defined as paths (connecting boundaries) or loops of Pauli operators on the lattice (see Figure~\ref{fig:surface-codes}). For smooth defects, logical operators are
\( X_L = \prod_{e \in P_X} X_e, \quad Z_L = \prod_{e \in P_Z} Z_e \),
where \(P_X\) is a path of \(X\)-operators connecting two smooth defects, and \(P_Z\) is a loop of \(Z\)-operators encircling a smooth defect. For rough defects, the roles of \(X\) and \(Z\) are reversed: \(X_L\) encircles a rough defect, and \(Z_L\) connects two rough defects. Logical operators correspond to chains of errors undetectable by the code. Boundaries also can be rough(smooth), supporting undetected \(Z(X)\)-chains.

The \textbf{dual lattice} has a vertex for each plaquette of the primal lattice and an edge for each primal edge. Stabilisers \(A_v\) on the primal lattice thus correspond to stabilisers \(B_p\) on the dual lattice, and logical operators map between the primal and dual lattices. This duality ensures symmetry between \(X\)- and \(Z\)-type errors, reinforcing the fault tolerance of the surface code.

Visualising \textbf{braiding} of logical operators is easiest in a spacetime view (although see also~\cite[\S V]{Fowler2009} for clear step-by-step diagrams of a logical CNOT in the 2D view). We construct a 3D lattice, \(\mathcal{L}\), with physical qubits placed on both the edges and faces of the lattice. This can be viewed as a \textit{measurement-based quantum computation (MBQC)} simulation of the 2D surface code~\cite{raussendorf2006fault,austingoyal2009}, with the third lattice dimension as simulated time. The 3D cluster state is initialised by preparing qubits on the edges and faces of \(\mathcal{L}\) in the \(|+\rangle\) state and applying controlled-X gates between face qubits and edge qubits along their boundaries. Stabilisers for this state are given by:

\begin{equation*}
	K(c_2) = X(c_2) Z(\partial c_2) \quad  \text{ and } \quad  K(\bar{c}_2) = X(\bar{c}_2) Z(\partial \bar{c}_2),
\end{equation*}

where \(c_2 \in C_2\) corresponds to primal faces, and \(\bar{c}_2 \in \bar{C}_2\) corresponds to dual faces. Logical qubits are encoded in defects that extend through the 3D lattice, and logical operators are defined in terms of these defects. For primal defects, the logical \(Z_L\) operator is the boundary of a dual 2-chain, \(Z_L = Z(\partial \bar{c}_2)\), and the logical \(X_L\) operator is defined along a primal 1-chain \(c_1\) as \(X_L = X(c_1)\).

The 3D cluster is divided into three regions. The \textit{defect region} (\(D = D_p \cup D_d\)) includes both primal defects (\(D_p\)) and dual defects (\(D_d\)) and represents areas where stabilisers are deliberately removed. The \textit{vacuum region} (\(V = V_p \cup V_d\)) contains the rest of the lattice where stabilisers are intact, including primal (\(V_p\)) and dual (\(V_d\)) vacuum regions. Finally, the \textit{singular qubit region} (\(S\)) includes qubits at the interface of defects and vacuum. A specific measurement-based computation is performed based on four cluster areas.
Qubits along defect chains (\(d \subset D\)) are measured in the \(Z\)-basis. Qubits on defect faces whose boundaries intersect \(d\) are measured in the \(X\)-basis.
Qubits in the vacuum region (\(V\)) are measured in the \(X\)-basis. Finally, qubits in the singular region (\(S\)) are measured in the \(\frac{X \pm Y}{\sqrt{2}}\)-basis.

To ensure compatibility between the stabiliser structure and the measurement pattern, the stabilisers must commute with the measurements:
\begin{equation*}
	\left[K(c_2) K(\bar{c}_2), X_a\right] = 0 \text{ for all } a \in V, \quad  \left[K(c_2) K(\bar{c}_2), Z_b\right] = 0 \text{ for all } b \in D.
\end{equation*}
These conditions imply that for any primal 2-chain \(c_2 \in C_2\) and dual 2-chain \(\bar{c}_2 \in \bar{C}_2\):
\begin{equation*}
 \{c_2\} \subset V_p, \quad \{\partial c_2\} \subset D_p, \quad \{\bar{c}_2\} \subset V_d, \quad \{\partial \bar{c}_2\} \subset D_d.
\end{equation*}

Within this 3D description, a primal defect carries
\(
Z_L=Z(\partial\bar c_2)
\)
and
\(
X_L=X(c_1),
\)
where \(\partial\bar c_2\) encircles the defect and \(c_1\) connects time-separated cross-sections. Braiding a smooth defect (dual defect) around a rough defect (primal defect) induces transformations of the logical operators as the stabilisers of the correlation surfaces are multiplied at the point of braiding and then separated again (analogous to the merge and splitting operations of lattice surgery):
\begin{equation*}
	X_L^\text{(smooth)} \to X_L^\text{(smooth)} X_L^\text{(rough)} \quad \text{ and } \quad Z_L^\text{(rough)} \to Z_L^\text{(smooth)} Z_L^\text{(rough)}.
\end{equation*}
This corresponds to the application of a \gatename{CNOT} gate, where the smooth defect is the control qubit and the rough defect is the target qubit. Gates can be induced between logical qubits of the same type by introducing intermediate braids.\footnote{As every operation to merge, split, move, introduce, or remove defects involve changes to the stabilisers being measured, there will also be uncontrollable but \emph{heralded} byproduct operations, of the sort familiar from MBQC~\cite{rossmbqc, will-simmons-flow} and lattice surgery~\cite{deBeaudrap2020zxcalculusis, bombin2023unifying}.
Strictly speaking, these must be accounted for to describe general braiding operations as CPTP maps.
However, we follow convention in acknowledging and then dropping explicit mention of them, as a first step in the development of a workflow for a more comprehensive and systematic approach to analysing defect braiding procedures.}

Informal `topological calculi' diagrams have been used since the beginning to illustrate the complex procedures of braiding, and to aid calculation of stabilisers~\cite{raussendorf2006fault,raussendorf2007topological}. Later work used, again informally, computer-aided 3D diagram generation~\cite{Gidney2020Video,paler2018specification,paler2017fault,fowler2012bridge}. However, reasoning about braids has remained cumbersome, and lacks complete formal tools for circuit design and compilation.

\vspace{-0.4cm}

\section{Braiding of defects as a graphical calculus: the category $\catname{KNOT}$}
\label{sec:braiding}

We represent maps in $\catname{KNOT}$ by diagrams, composed from generators as below.
These depict braiding of defects in a surface code, in the form of a 2D diagram with additional information for wires crossing over or under other wires, in the manner of knot diagrams.
\begin{definition}
	The category $\catname{KNOT}$ (also called the \textit{category of defects}) is a monoidal category whose objects are bit-strings, tensor product is given by the concatenation of bit-strings,  morphisms are generated by the below diagrams and for domains and codomains of morphisms $0$ is represented by a grey wire and $1$ by a black wire.

    \begin{center}
		\noindent
			\hspace*{3.5em}
			\begin{tabular}{r@{\qquad}l}
			\begin{tabular}{c}
				\textbf{Identity wires} \\
				\begin{tabular}{cc}
					$\begin{tikzpicture}[scale=0.65]
	\begin{pgfonlayer}{nodelayer}
		\node [style=none] (0) at (-1, 0) {};
		\node [style=none] (1) at (1, 0) {};
	\end{pgfonlayer}
	\begin{pgfonlayer}{edgelayer}
		\draw [style=primal] (0.center) to (1.center);
	\end{pgfonlayer}
\end{tikzpicture}
$
					&
					$\begin{tikzpicture}[scale=0.65]
	\begin{pgfonlayer}{nodelayer}
		\node [style=none] (0) at (-1, 0) {};
		\node [style=none] (1) at (1, 0) {};
	\end{pgfonlayer}
	\begin{pgfonlayer}{edgelayer}
		\draw [style=dual] (0.center) to (1.center);
	\end{pgfonlayer}
\end{tikzpicture}
$
				\end{tabular}
			\end{tabular}
			&
			\begin{tabular}{c}
				\textbf{Empty bit-strings} \\
				\begin{tabular}{cccc}
					$\begin{tikzpicture}[scale=0.65]
	\begin{pgfonlayer}{nodelayer}
		\node [style=none] (0) at (-1, 0) {};
		\node [style=none] (1) at (1, 0) {};
	\end{pgfonlayer}
	\begin{pgfonlayer}{edgelayer}
		\draw [style={dual-dot}] (0.center) to (1.center);
	\end{pgfonlayer}
\end{tikzpicture}
$
					\!\!&\!\!
					$\begin{tikzpicture}[scale=0.65]
	\begin{pgfonlayer}{nodelayer}
		\node [style=none] (0) at (-1, 0) {};
		\node [style=none] (1) at (1, 0) {};
	\end{pgfonlayer}
	\begin{pgfonlayer}{edgelayer}
		\draw [style={dual-dot}] (1.center) to (0.center);
	\end{pgfonlayer}
\end{tikzpicture}
$
					&
					$\begin{tikzpicture}[scale=0.65]
	\begin{pgfonlayer}{nodelayer}
		\node [style=none] (0) at (-1, 0) {};
		\node [style=none] (1) at (1, 0) {};
	\end{pgfonlayer}
	\begin{pgfonlayer}{edgelayer}
		\draw [style={primal-dot}] (0.center) to (1.center);
	\end{pgfonlayer}
\end{tikzpicture}
$
					\!\!&\!\!
					$\begin{tikzpicture}[scale=0.65]
	\begin{pgfonlayer}{nodelayer}
		\node [style=none] (0) at (-1, 0) {};
		\node [style=none] (1) at (1, 0) {};
	\end{pgfonlayer}
	\begin{pgfonlayer}{edgelayer}
		\draw [style={primal-dot}] (1.center) to (0.center);
	\end{pgfonlayer}
\end{tikzpicture}
$
				\end{tabular}
			\end{tabular}
			\end{tabular}
			\\[2ex]
		\noindent
			\hspace*{-1.5em}
			\begin{tabular}{c@{\,\,}c@{\,\,}c@{\,\,}c}
				\begin{tabular}{c}
					\textbf{Same-type crossings} \\[5pt]
					\begin{tabular}{c@{\;\;}c@{\;\;}c}
						$\input{knot-generators-tikzit/7.tikz}$
						\!&\!
						$\input{knot-generators-tikzit/8.tikz}$
						\!&\!
						$\input{knot-generators-tikzit/9.tikz}$
					\\[1ex]
						$\input{knot-generators-tikzit/10.tikz}$
						\!&\!
						$\input{knot-generators-tikzit/11.tikz}$
						\!&\!
						$\input{knot-generators-tikzit/12.tikz}$
					\end{tabular}
				\end{tabular}
				&
				\begin{tabular}{c}
					\textbf{Opposite-type crossings} \\[5pt]
					\begin{tabular}{c@{\;\;}c}
						$\input{knot-generators-tikzit/13.tikz} = \input{knot-generators-tikzit/14.tikz}$
					\\[1ex]
						$\input{knot-generators-tikzit/15.tikz} = \input{knot-generators-tikzit/16.tikz}$
					\end{tabular}
				\end{tabular}
				&
				\begin{tabular}{c}
					\textbf{Bifurcation maps} \\[5pt]
					\begin{tabular}{c}
						$\input{knot-generators-tikzit/19.tikz}$
				  \\[1ex]
						$\input{knot-generators-tikzit/21.tikz}$
					\end{tabular}
				\end{tabular}
				&
				\begin{tabular}{c}
					\textbf{Pauli decorators}  \\[5pt]
					\begin{tabular}{ccc}
						$\input{knot-generators-tikzit/27.tikz}$
						\!&\!
						$\input{knot-generators-tikzit/K19-R.tikz}$
						\!&\!
						$\begin{tikzpicture}[scale=0.65]
	\begin{pgfonlayer}{nodelayer}
		\node [style=none] (0) at (-1, 0) {};
		\node [style=none] (1) at (1, 0) {};
		\node [style={red-dot}] (2) at (0, 0) {};
	\end{pgfonlayer}
	\begin{pgfonlayer}{edgelayer}
		\draw [style=dual] (0.center) to (1.center);
	\end{pgfonlayer}
\end{tikzpicture}
$
					\\[1ex]
						$\input{knot-generators-tikzit/28.tikz}$
						\!&\!
						$\input{knot-generators-tikzit/31.tikz}$
						\!&\!
						$\begin{tikzpicture}[scale=0.65]
	\begin{pgfonlayer}{nodelayer}
		\node [style=none] (0) at (-1, 0) {};
		\node [style=none] (1) at (1, 0) {};
		\node [style={red-dot}] (2) at (0, 0) {};
	\end{pgfonlayer}
	\begin{pgfonlayer}{edgelayer}
		\draw [style=primal] (0.center) to (1.center);
	\end{pgfonlayer}
\end{tikzpicture}
$
					\end{tabular}
				\end{tabular}
			\end{tabular}
	  \end{center}

    \smallskip
	\noindent
	Diagrams are read from left to right, representing a progression from the input to the output.
	A grey or a black dot on the right (left) end of a diagram represents an empty bit-string as a domain (codomain); similarly, cups and caps have empty bit-strings as domains or codomains. The generators satisfy the equivalences shown in Table~\ref{table:KNOT}.

  \begin{table}[p]
    \centering
    \renewcommand{\arraystretch}{1.5}
    \noindent
	\begin{minipage}{0.975\textwidth}\raggedright
		Planar isotopy rules:
	\end{minipage}
	\vspace*{4mm}
	\\
    \begin{minipage}{0.575\textwidth}\raggedleft
        \begin{minipage}[t]{0.6\textwidth}
            \begin{align}
                \tag{K1} \input{knot-generators-tikzit/7.tikz} = \input{knot-generators-tikzit/8.tikz} &= \input{knot-generators-tikzit/9.tikz} \label{eq:K1}
            \\[2ex]
                \tag{K2} \input{knot-generators-tikzit/K3-l.tikz} &= \input{knot-generators-tikzit/K3-r.tikz} \label{eq:K3}
            \end{align}
        \end{minipage}
        \hfill
        \begin{minipage}[t]{0.375\textwidth}
            \begin{align}
                \tag{K3} \input{knot-generators-tikzit/K2-l.tikz} &= \input{knot-generators-tikzit/K2-r.tikz} \label{eq:K2}
			\\[3ex]
				\tag{K4} \input{knot-generators-tikzit/K12-r.tikz} &= \input{knot-generators-tikzit/K12-l.tikz} \label{eq:K12}
            \end{align}
        \end{minipage}
		\\
		\begin{minipage}{1.0\textwidth}\raggedright
			\begin{equation}
				\tag{K5}
				\input{knot-generators-tikzit/K24-L.tikz} = \input{knot-generators-tikzit/K24-R.tikz}
				\label{eq:loop-D-commute}
			\end{equation}
		\end{minipage}
    \end{minipage}
    \hfill
    \begin{minipage}{0.4\textwidth}\raggedright
            \begin{align}
				\tag{K6} \input{knot-generators-tikzit/K21-L.tikz} &=
					\input{knot-generators-tikzit/K21-R.tikz} \label{eq:overcross-kill}
			\\[2ex]
                \tag{K7} \input{knot-generators-tikzit/K4-l.tikz} &= \input{knot-generators-tikzit/K4-r.tikz} \label{eq:K4}
			\\[3ex]
				\tag{K8} \input{knot-generators-tikzit/K7-l.tikz} &= \input{frob/48.tikz} \label{eq:K5}
            \end{align}
    \end{minipage}
	\\[4ex]
	\vspace*{4mm}
	\begin{minipage}{0.975\textwidth}\raggedright
		Base rules:
	\end{minipage}
	\\
    \begin{minipage}{0.575\textwidth}\raggedleft
		\begin{align}
			\tag{K9} \input{knot-generators-tikzit/K8-l.tikz} &= \input{knot-generators-tikzit/K8-r.tikz} \label{eq:loop-around-one-wire}
	\\[1ex]
			\tag{K10} \input{knot-generators-tikzit/K5-l.tikz} &= \input{knot-generators-tikzit/K5-r.tikz} \label{eq:K6}
		\end{align}
    \end{minipage}
    \hfill
    \begin{minipage}{0.4\textwidth}\raggedright
            \begin{align}
				\tag{K11} \input{knot-generators-tikzit/K14-L.tikz} &= \input{knot-generators-tikzit/K14-R.tikz} \label{eq:K-spider}
			\\[0ex]
				\tag{K12} \input{knot-generators-tikzit/K22-L.tikz} &= \input{knot-generators-tikzit/K22-R.tikz} \label{eq:spider-disappears}
            \end{align}
    \end{minipage}
	\vspace*{4mm}
	\\[0.5ex]
	\begin{minipage}{0.975\textwidth}\raggedright
		Decorator rules:
	\end{minipage}
	\vspace*{4mm}
	\\
	\begin{minipage}{0.6\textwidth}\raggedleft
		\begin{minipage}[t]{0.55\textwidth}
			\begin{align}
				\tag{K13} \input{knot-generators-tikzit/K10-l.tikz} &= \input{knot-generators-tikzit/K10-r.tikz} \label{eq:K9}
				\vspace*{5mm}
			\\[3ex]
				\tag{K14} \input{knot-generators-tikzit/K6-l.tikz} &= \input{knot-generators-tikzit/K6-r.tikz} \label{eq:K10}
			\end{align}
		\end{minipage}
		\hfill
		\begin{minipage}[t]{0.375\textwidth}
			\begin{align}
				\tag{K15} \input{knot-generators-tikzit/K15-L.tikz} &= \input{knot-generators-tikzit/K15-R.tikz} \label{eq:K-bar-slide}
				\vspace*{3mm}
			\\[2.5ex]
				\tag{K16} \input{knot-generators-tikzit/K16-L.tikz} &= \input{knot-generators-tikzit/K16-R.tikz} \label{eq:K-circle-slide}
			\end{align}
		\end{minipage}
		\\[0.1ex]
		\begin{minipage}{0.95\textwidth}
			\begin{align}
				\tag{K17} \input{knot-generators-tikzit/K17-L.tikz} = \input{knot-generators-tikzit/K17-M.tikz} &= \input{knot-generators-tikzit/K17-R.tikz} = \dots \label{eq:K-bar-spider}
			\\[0.3ex]
				\tag{K18} \input{knot-generators-tikzit/K9-L.tikz} &= \input{knot-generators-tikzit/K9-R.tikz} \label{eq:K7}
			\\ \notag
			\end{align}
		\end{minipage}
	\end{minipage}
	\hfill
	\begin{minipage}{0.375\textwidth}\raggedright
			\begin{align}
				\tag{K19} \input{knot-generators-tikzit/K18-L.tikz} &= \input{knot-generators-tikzit/K18-R.tikz} \label{eq:K-circle-spider}
			\\[0ex]
				\tag{K20} \input{knot-generators-tikzit/K19-L.tikz} &= \input{knot-generators-tikzit/K19-R.tikz} \label{eq:K-bar-upside}
			\\[0ex]
				\tag{K21} \input{knot-generators-tikzit/K20-L.tikz} &= \input{knot-generators-tikzit/K20-R.tikz} \label{eq:K-circle-falls}
			\\[0ex]
				\tag{K22} \input{knot-generators-tikzit/K23-L.tikz} &= \input{knot-generators-tikzit/K23-R.tikz} \label{eq:K-circle-bar-commute}
			\\ \notag
			\end{align}
	\end{minipage}
	\caption{%
	\label{table:KNOT}%
	The rules of the $\catname{KNOT}$ calculus. \ref{eq:loop-D-commute}: $D$ can be any diagram in $\catname{KNOT}$ without any red decorators and have arbitrary input and output wires of the same colour as the loops around the wires. We also require $p, q \in \{0, 1, 2\}$.
	\textit{The planar isotopy rules} encode (some of) the rules of planar isotopy for opposite-type wires only. Same-type wires do not interact with each other. \textit{The base rules} encode the equalities which are not covered by planar isotopy.
	}
	\vspace*{3ex}
	\end{table}

\end{definition}

Diagrams of $\catname{KNOT}$ represent the braiding of defects in a surface code, and correlation surfaces (spacetime stabilisers) supported on those defects.%
    \footnote{%
        The effect of those operations, supporting the equivalences of Table~\ref{table:KNOT}, are standard in the literature on defect braiding~\cite{raussendorf2007topological}.
        For the sake of brevity, we defer any account of these effects to our presentation of a map from $\catname{KNOT}$ into ZX-diagrams.}

\vspace{-0.2cm}
\paragraph{Correlation surfaces.}

In the literature, the red decorators above are often omitted as they represent operations unachievable by braiding and can be seen as undetectable errors. However, logical transformations from braiding can be described by rewriting diagrams to ``move" red generators from input to output. Correlation surfaces in the surface code~\cite{raussendorf2006fault,raussendorf2007topological} track the evolution of logical operators or errors as the code’s topology changes. They geometrically encode how logical operators—initially represented by chains or loops—are deformed by boundary shifts, defect braiding, or lattice manipulations. This tracking constructs space-time stabilisers, summarizing logical transformations and showing how an initial operator, such as \( X_{\text{in\;$1$}} \), evolves into \( X_{\text{out\;$1$}} X_{\text{out\;$2$}} \) via a logical map \( M \) (the action of a \gatename{CNOT} gate).

\vspace{-0.2cm}
\paragraph{Soundness.}

A basic property of a rewrite system with a semantic map should be the soundness for its interpretation. The informal diagrams given in e.g. Ref.~\cite{raussendorf2007topological} can be interpreted as describing a specific family of surface code computation procedures described by a braiding pattern. The goal of such a procedure is to perform a map on the logical operators (or non-correctable errors) from the start to the end of the computation (to define a so-called space-time stabiliser). Thus, one can instead give an interpretation in the maps on the space of logical operators. This is what we also do in this paper. We are not concerned with the physical implementation of a specific braiding pattern as a set of operations on the surface code. We are instead interested in the space-time stabilisers that could be implemented by such procedures. We discuss this more in the Appendix~\ref{section:proof-soundness}.

\begin{theorem}
	\label{soundness-maps}
	$\catname{KNOT}$ is sound for the logical maps implemented by the braided surface code error correction procedure up to classical byproducts.

	\hfill
	\textup{[Proof in Appendix \ref{section:proof-soundness}.]}
\end{theorem}

Different boundary conditions lead to different sets of equivalences for braiding. For all defects, we allow the surfaces to connect a defect to a boundary, encircle the defects, and connect two defects of the same type. A correlation surface connecting two defects of the same type, is equivalent to two independent correlation surfaces connecting the defects to a common boundary of the appropriate type; we freely use this equivalence to simplify our analysis.

\begin{example}
	This example is an analogue of post-selected teleportation. In the next sections, we will explain how the generators can be associated with Hilbert spaces.
	\begin{allowdisplaybreaks}
			\begin{align*}
				\input{teleportation/1.tikz} &\overset{\eqref{eq:K7}}{=}
				\input{teleportation/2.tikz} \overset{\substack{\eqref{eq:spider-disappears} \\ \eqref{eq:K-circle-spider}\\ \eqref{eq:K-bar-spider}}}{=}
				\input{teleportation/3.tikz} \overset{\substack{\eqref{eq:K9} \\ \eqref{eq:K7}}}{=} \\[20pt]
				\input{teleportation/4.tikz} &\overset{\substack{\eqref{eq:spider-disappears} \\ \eqref{eq:K-circle-spider}\\ \eqref{eq:K-bar-spider}}}{=}
				\input{teleportation/5.tikz} \overset{\substack{\eqref{eq:K9} \\ \eqref{eq:K10}}}{=}
				\input{teleportation/6.tikz} \overset{\substack{\eqref{eq:overcross-kill} \\ \eqref{eq:K5}}}{=} \\[20pt]
				\input{teleportation/7.tikz} & =
				\input{teleportation/8.tikz} =
				\input{teleportation/9.tikz} \overset{\eqref{lemma:ring-around}}{=}
				\input{teleportation/10.tikz}
			\end{align*}
		\end{allowdisplaybreaks}
\end{example}

\vspace{-0.4cm}

\section{Interpreting the diagrams}
\label{sec:interpretation}
In this section, we study the meaning of the $\catname{KNOT}$ diagrams and provide interpretations. We also show how, by applying a doubling construction we may recover the standard defect-pair encoding: in this setting, we describe two sublanguages generated by commonplace braiding techniques, and show that they are each sound and complete for the (0, $\pi$)-fragment of ZX.

\vspace{-0.2cm}
\subsection{A partial logical semantics of $\catname{KNOT}$ via the ZX-calculus}
Thus far, we have only indicated that $\catname{KNOT}$ diagrams are intended to denote the logical effects of the corresponding braiding operations ignoring the byproduct operations.
We have not said what these effects are, or proven the soundness of the rewrites (K1)--(K22): this sort of denotation is common in the literature, but relies on knowledge of what those effects are.
One of our contributions is to describe the logical effect of braiding operations, by providing semantics for $\catname{KNOT}$ via ZX-diagrams.
We prove the correctness of these semantics in the Appendix, but for now they serve as exposition.

We proceed by defining a map $z$ from the generators of $\catname{KNOT}$ to ZX-diagrams  as follows (where we group together generators which have the same semantics for brevity): we present this in Table~\ref{table:KNOT-ZX-semantics}.

\begin{table}[h]
    \centering
    \vspace*{-3ex}
\noindent
~\hfill
\begin{minipage}{0.2\textwidth}
\begin{equation*}
    \left.
    \begin{aligned}
        \input{knot-generators-tikzit/1.tikz}
        \\[1ex]
        \input{knot-generators-tikzit/2.tikz}
    \end{aligned}
    \;\right\}
    \,\xmapsto{z}\;
    \tikz[tikzfig]{ \draw (0,0) -- (3,0); }
\end{equation*}
\end{minipage}
~\hfill~
\begin{minipage}{0.15\textwidth}
\begin{equation*}
    \left.
    \begin{aligned}
        \input{knot-generators-tikzit/23.tikz}
        \\[1ex]
        \input{knot-generators-tikzit/26.tikz}
    \end{aligned}
    \;\right\}
    \,\xmapsto{z}
    \input{zx/cup.tikz}
\end{equation*}
\end{minipage}
~\hfill~
\begin{minipage}{0.15\textwidth}
\begin{equation*}
    \left.
    \begin{aligned}
        \input{knot-generators-tikzit/24.tikz}
        \\[1ex]
        \input{knot-generators-tikzit/25.tikz}
    \end{aligned}
    \;\right\}
    \;\xmapsto{z}
    \input{zx/cap.tikz}
\end{equation*}
\end{minipage}
~\hfill~
\begin{minipage}{0.4\textwidth}
\begin{equation*}
    \left.
    \begin{gathered}
        \input{knot-generators-tikzit/7.tikz}  \,,
        \;
        \input{knot-generators-tikzit/8.tikz}  \,,
        \;
        \input{knot-generators-tikzit/9.tikz}
        \\[1ex]
        \input{knot-generators-tikzit/10.tikz}  \,,
        \;
        \input{knot-generators-tikzit/11.tikz}  \,,
        \;
        \input{knot-generators-tikzit/12.tikz}
    \end{gathered}
    \;\right\}
    \,\xmapsto{z}
    \input{zx/swap.tikz}
\end{equation*}
\end{minipage}
\hfill~
\vspace*{-1ex}

\begin{align*}
	\input{knot-generators-tikzit/3.tikz} \;\xmapsto{z}{}&\, \input{zx/green-effect.tikz}
	\qquad&
	\input{knot-generators-tikzit/4.tikz} \;\xmapsto{z}{}&\, \input{zx/green-state.tikz}
	\qquad&
	\input{knot-generators-tikzit/5.tikz} \;\xmapsto{z}{}&\, \input{zx/red-effect.tikz}
\\[1ex]
	\input{knot-generators-tikzit/19.tikz} \;\xmapsto{z}{}& \input{zx/red-merge.tikz}
	\qquad&
	\input{knot-generators-tikzit/21.tikz} \;\xmapsto{z}{}& \input{zx/green-merge.tikz}
	\qquad&
	\input{knot-generators-tikzit/6.tikz}  \;\xmapsto{z}{}&\, \input{zx/red-state.tikz}
\end{align*}

\vspace*{-2ex}

\noindent
~
\hfill
\begin{minipage}{0.2\textwidth}
\begin{equation*}
    \left.
    \begin{aligned}
        \input{knot-generators-tikzit/27.tikz}
        \\[1ex]
        \input{knot-generators-tikzit/29.tikz}
    \end{aligned}
    \;\right\}
    \,\xmapsto{z}
    \input{zx/red-pi.tikz}
\end{equation*}
\end{minipage}
~\hfill~
\begin{minipage}{0.2\textwidth}
\begin{equation*}
    \left.
    \begin{aligned}
        \input{knot-generators-tikzit/30.tikz}
        \\[1ex]
        \input{knot-generators-tikzit/28.tikz}
    \end{aligned}
    \;\right\}
    \,\xmapsto{z}
    \input{zx/green-pi.tikz}
\end{equation*}
\end{minipage}
~\hfill~
\begin{minipage}{0.25\textwidth}
\begin{equation*}
    \left.
    \begin{aligned}
        \input{knot-generators-tikzit/13.tikz}
        \\[1ex]
        \input{knot-generators-tikzit/14.tikz}
    \end{aligned}
    \;\right\}
    \,\xmapsto{z}
    \input{zx/z-cnot.tikz}
\end{equation*}
\end{minipage}
~\hfill~
\begin{minipage}{0.25\textwidth}
\begin{equation*}
    \left.
    \begin{aligned}
        \input{knot-generators-tikzit/15.tikz}
        \\[1ex]
        \input{knot-generators-tikzit/16.tikz}
    \end{aligned}
    \;\right\}
    \,\xmapsto{z}
    \input{zx/x-cnot.tikz}
\end{equation*}
\end{minipage}
\hfill
\hfill
~
    \caption{%
        \label{table:KNOT-ZX-semantics}%
		Semantics for $\catname{KNOT}$ diagram generators, in terms of ZX diagrams. These semantics represent the effect of the corresponding braiding procedures (in the absence of byproduct operations), on qubits encoded in correlation surfaces. Note that, as a consequence, these maps are not necessarily norm-preserving.}
    \vspace*{-1ex}
    \end{table}

\begin{proposition}
  \label{prop:z-functor-zx}
	The map $z$ lifts to a functor $Z: \catname{KNOT} \rightarrow \catname{ZX}$.

	\hfill
	\textup{[Proof in Appendix \ref{proof:z-functor-zx}]}
\end{proposition}

\begin{corollary}
	$\catname{KNOT}$ is sound for the ($0$, $\pi$)-fragment of ZX-calculus via the functor $Z$.
\end{corollary}

\noindent
Thus, we may interpret and analyse $\catname{KNOT}$ with reference to (a fragment of) the ZX-calculus, via the functor $Z$.
Mappings similar to ours have appeared in the literature (e.g., \cite{bombin2023unifying, Horsman-2011,paler2017fault, hanks2020effective}), but we provide a formal proof that these rules are sound with respect to the maps they implement.

\vspace*{-2ex}
\paragraph{Remark.} Similar expositions are often based on the assumption of an infinite surface code patch, which lacks boundaries (e.g., see \cite{raussendorf2006fault,raussendorf2007topological}).
In contrast, we assume that the patch includes both boundary types at every time slice of the error correction procedure. This assumption introduces some changes to the semantics and syntax of the diagrams if we treat one defect as representing a qubit. If we treat a defect pair as representing a qubit, the usual semantics of the braided diagrams from the literature can be recovered as we will see below.

\vspace{-0.2cm}
\subsection{Double defects}
As we describe below, braiding individual defects to encode logical operations has some practical limitations.
A standard remedy is to use a further `defect pair' encoding for logical qubits~\cite{Horsman-2011,raussendorf2006fault,raussendorf2007topological,paler2016synthesis}. With a single wire representing a qubit, a \gatename{CNOT} gate can only be performed with a black `primal' corresponding to the control qubit, and a gray `dual' wire corresponding to the target, using the mapping $z$ defined as in Table~\ref{table:KNOT-ZX-semantics} (see Figure~\ref{fig:cnots-singls}).

\begin{figure}[!t]
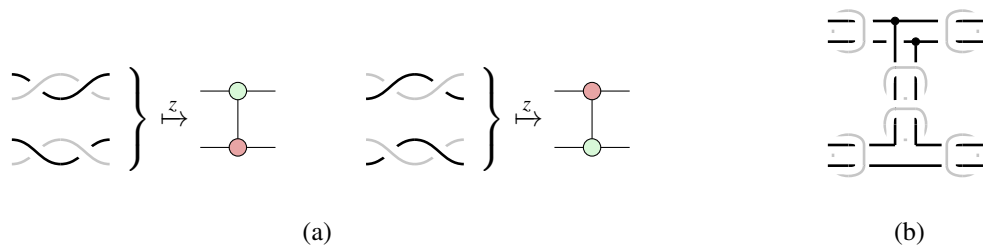

	\centering
	\begin{subfigure}[b]{0.65\textwidth}
		\begin{align*}
			\left.
			\begin{aligned}
				\input{knot-generators-tikzit/13.tikz}
				\\[1ex]
				\input{knot-generators-tikzit/14.tikz}
			\end{aligned}
			\;\right\}
			\,\xmapsto{z}{}&
			\input{zx/z-cnot.tikz}
			&
			\left.
			\begin{aligned}
				\input{knot-generators-tikzit/15.tikz}
				\\[1ex]
				\input{knot-generators-tikzit/16.tikz}
			\end{aligned}
			\;\right\}
			\,\xmapsto{z}{}&
			\input{zx/x-cnot.tikz}
		\end{align*}
	\caption{}
	\label{fig:cnots-singls}
	\end{subfigure}
	\hspace{0.03\textwidth}
	\begin{subfigure}[b]{0.25\textwidth}
		\begin{equation*}
			\input{zx-knot/cnot-doubled.tikz}.
		\end{equation*}
	\caption{}
	\label{fig:cnots-doubled}
	\end{subfigure}
	\caption{Two ways of performing a \gatename{CNOT} gate. (a) Using a single defect, the control is always a black wire and the target is always a gray wire. (b) A \gatename{CNOT} gate using a pair of defects -- the gate can be performed in either direction.}
	\label{fig:cnots-example}
\end{figure}

A solution to this is offered in Ref.~\cite{raussendorf2007topological} (albeit motivated in a somewhat different context).
This solution is to encode logical qubits in the logical operators of a pair of defects, and to realise \gatename{CNOT} operations in this encoding with a braiding procedure that we may denote in $\catname{KNOT}$ as shown in Figure~\ref{fig:cnots-doubled}.

This construction enables \gatename{CNOT} gates to be performed in both directions, as the roles of control and target are encoded in the geometry of the procedure rather than the defect type.
This technique requires a different encoding from what is described in Table~\ref{table:KNOT-ZX-semantics}, but which we may describe in relation to the semantics provided by the functor $Z$ described above.

We now formalise and generalise this construction, in the setting of a finite lattice with external boundaries, by providing semantics for the doubled construction in the ZX-calculus in which pairs of defects correspond to a single wire. The category is similar to $\text{CPM}[\catname{KNOT}]$ with the addition of cross-caps for arbitrary wires.
In our subcategory, the black wires will be treated as `qubit' wires, and the grey wires as mediating the interactions between the black wires. This corresponds to a convention in compiling procedures to defect braiding~\cite{paler2014design,paler2016synthesis,paler2017fault,paler2018specification,paler2019surfbraid}, in which `primal' defects are the main defects and the dual defects serve only to help realise logical transformations.

We refer to the structures of eqs. \eqref{duality-1} and \eqref{duality-2} on the LHS to be ribbon-like and the structures on the RHS to be tangle-like. Below, we describe two subtheories in which these different types of structures are generators, and show that each are complete for the ($0$, $\pi$)-fragment of ZX.

\vspace{-0.2cm}
\subsubsection{Tangle-like and ribbon-like structures}

\begin{definition}
	Define $\catname{KNOT}_{\text{doubled}}$ to be a subcategory of $\catname{KNOT}$ with morphisms generated by the following doubled diagrams:

	\hspace*{-16.5ex}
	\begin{equation}
		\label{duality-1}
		\input{lemmas/lemma2-r-b-copy.tikz} \overset{\eqref{lemma-duality}}{=} \input{lemmas/lemma2-l-copy.tikz}
	\end{equation}
	\begin{equation}
		\label{duality-2}
		\input{lemmas/lemma-3-r-copy.tikz} \overset{\eqref{lemma-duality}}{=} \input{lemmas/lemma-3-l-copy.tikz}
	\end{equation}
	\vspace*{-10ex}
	\\
	\begin{center}
	\begin{equation}
		\label{swap-duality}\input{knot-doubled/swap.tikz} \overset{\eqref{lemma:ring-around}}{=} \input{knot-doubled/swap-b.tikz}
	\end{equation}
	\end{center}
	where the red decorators are optional and we also quotient by the precise placement of the red decorators on inner loops (whose exact position is negotiable through the rule~\eqref{eq:K7}). The numerical indices for the wires and the dashed lines are meant only to help the reader and are not a part of the theory. Two diagrams are considered equal if they can be transformed into one another by the rules of $\catname{KNOT}$ (possibly passing through immediate steps which are not in $\catname{KNOT}_{\text{doubled}}$).
\end{definition}
\noindent

All six of these generators are generalisations of patterns that are used to describe braiding procedures \cite{raussendorf2007topological, paler2014design}. However, it is possible to show (see Lemma~\ref{lemma-duality}) using only rewrites of $\catname{KNOT}$ that the left-hand side generators can be rewritten into the right-hand side generators. By the soundness of $\catname{KNOT}$ for $\catname{ZX}$ through $Z$, we have that the ZX-diagrams corresponding to the diagrams on the right-hand side can be rewritten to the ZX-diagrams corresponding to the diagrams on the right-hand side. We can therefore reduce our analysis of this category to a treatment of the right-hand side generators; we do so below.

A number of interesting identities hold in $\catname{KNOT}_{\text{doubled}}$. For example, we have the following:

\begin{allowdisplaybreaks}
\hspace{-1.3cm}
\hspace*{-10ex}
\begin{equation}
	\label{spider-1-knot}
	\input{zx-knot/spider1-example.tikz} = \input{zx-knot/lemma2-l-copy.tikz} \end{equation}

\begin{equation}
	\label{spider-2-knot}
	\input{zx-knot/spider2-example.tikz} = \input{zx-knot/lemma-3-l-copy.tikz}\end{equation}

\end{allowdisplaybreaks}
\noindent
If we translate the LHS of eq. \eqref{spider-1-knot} via the functor $Z$, we obtain the following equations in the ZX-calculus. The one below corresponds to the spider fusion law:
\begin{allowdisplaybreaks}
\begin{align}
	\input{zx-knot-spider/fusion1.tikz} &= \input{zx-knot-spider/fusion1-1.tikz} =
	\input{zx-knot-spider/fusion1-2.tikz} = \\ = \input{zx-knot-spider/fusion1-4.tikz} &=
	\input{zx-knot-spider/fusion1-5.tikz}
\end{align}
\end{allowdisplaybreaks}
\noindent
This is the equivalent of bialgebra law (for the $\catname{KNOT}$ diagram see eq.~\eqref{bialgebra}):
\begin{allowdisplaybreaks}
	\begin{align}
		\input{zx-knot-spider/bialgebra.tikz} &= \input{zx-knot-spider/bialgebra-1.tikz} =
		\input{zx-knot-spider/bialgebra-2.tikz} = \\ = \input{zx-knot-spider/bialgebra-2-1.tikz} &= \input{zx-knot-spider/bialgebra-3.tikz}
	\end{align}
	\end{allowdisplaybreaks}
\noindent

We give the full list of these identities in the Appendix (see Lemmas~\ref{lemma:green-spider-zx-knot}-\ref{lemma:circle-copies}). These lemmas might seems arbitrary but are crucial in observing the relationship between the ZX-calculus and $\catname{KNOT}_{\text{doubled}}$. In fact, eq. \eqref{bialgebra} is akin to the bialgebra law in the ZX-calculus, and the other equations are the spider and copy equations. The right-hand side of \eqref{green-spider-zx-knot} is the green spider of ZX-calculus and the right-hand side of \eqref{red-spider-zx-knot} is the red spider of ZX-calculus.

\begin{definition}
	Let $\catname{ZX}_{\text{doubled}}$ be the category of the ZX diagrams in the image of the functor $Z$ restricted to $\catname{KNOT}_{\text{doubled}}$.
\end{definition}

\begin{lemma}
	\label{lemma:braided-zx-gens}
	Explicitly, the morphisms of $\catname{ZX}_{\text{doubled}}$ are generated by:
	\vspace*{-1ex}
	\begin{equation*}
		\input{zx-knot-spider/red-spider-exploded.tikz} = \begin{tikzpicture}
	\begin{pgfonlayer}{nodelayer}
		\node [style=X dot] (0) at (0, 0) {$l \pi$};
		\node [style=Z dot] (1) at (-1.5, 1.5) {};
		\node [style=Z dot] (3) at (1.5, 1.5) {};
		\node [style=Z dot] (6) at (1.5, -1.5) {};
		\node [style=Z dot] (8) at (-1.5, -1.5) {};
		\node [style=none] (9) at (-2.25, 2) {};
		\node [style=none] (10) at (-2.25, 1) {};
		\node [style=none] (11) at (2.25, 2) {};
		\node [style=none] (12) at (2.25, 1) {};
		\node [style=none] (13) at (2.25, -1) {};
		\node [style=none] (14) at (2.25, -2) {};
		\node [style=none] (15) at (-2.25, -1) {};
		\node [style=none] (16) at (-2.25, -2) {};
		\node [style=none] (17) at (-1.25, 0) {$\vdots$};
		\node [style=none] (18) at (1.25, 0) {$\vdots$};
		\node [style=none] (19) at (-1.75, 0) {$m$};
		\node [style=none] (20) at (1.75, 0) {$n$};
	\end{pgfonlayer}
	\begin{pgfonlayer}{edgelayer}
		\draw [in=165, out=0] (1) to (0);
		\draw [in=-180, out=15] (0) to (3);
		\draw [in=0, out=-165] (0) to (8);
		\draw [in=180, out=-15] (0) to (6);
		\draw [in=180, out=30] (3) to (11.center);
		\draw [in=180, out=-15] (3) to (12.center);
		\draw [in=-180, out=15, looseness=1.25] (6) to (13.center);
		\draw [in=180, out=-15] (6) to (14.center);
		\draw [in=0, out=165] (1) to (9.center);
		\draw [in=0, out=-165] (1) to (10.center);
		\draw [in=0, out=165] (8) to (15.center);
		\draw [in=0, out=-165] (8) to (16.center);
	\end{pgfonlayer}
\end{tikzpicture}
 \qquad \qquad \input{zx-knot-spider/green-spider-exploded.tikz} = \begin{tikzpicture}
	\begin{pgfonlayer}{nodelayer}
		\node [style=Z dot] (0) at (0, 0) {$l \pi$};
		\node [style=Z dot] (1) at (-1.5, 1.5) {};
		\node [style=Z dot] (3) at (1.5, 1.5) {};
		\node [style=Z dot] (6) at (1.5, -1.5) {};
		\node [style=Z dot] (8) at (-1.5, -1.5) {};
		\node [style=none] (9) at (-2.25, 2) {};
		\node [style=none] (10) at (-2.25, 1) {};
		\node [style=none] (11) at (2.25, 2) {};
		\node [style=none] (12) at (2.25, 1) {};
		\node [style=none] (13) at (2.25, -1) {};
		\node [style=none] (14) at (2.25, -2) {};
		\node [style=none] (15) at (-2.25, -1) {};
		\node [style=none] (16) at (-2.25, -2) {};
		\node [style=none] (17) at (-1.25, 0) {$\vdots$};
		\node [style=none] (18) at (1.25, 0) {$\vdots$};
		\node [style=none] (19) at (-1.75, 0) {$m$};
		\node [style=none] (20) at (1.75, 0) {$n$};
	\end{pgfonlayer}
	\begin{pgfonlayer}{edgelayer}
		\draw [in=165, out=0] (1) to (0);
		\draw [in=-180, out=15] (0) to (3);
		\draw [in=0, out=-165] (0) to (8);
		\draw [in=180, out=-15] (0) to (6);
		\draw [in=180, out=30] (3) to (11.center);
		\draw [in=180, out=-15] (3) to (12.center);
		\draw [in=-180, out=15, looseness=1.25] (6) to (13.center);
		\draw [in=180, out=-15] (6) to (14.center);
		\draw [in=0, out=165] (1) to (9.center);
		\draw [in=0, out=-165] (1) to (10.center);
		\draw [in=0, out=165] (8) to (15.center);
		\draw [in=0, out=-165] (8) to (16.center);
	\end{pgfonlayer}
\end{tikzpicture}

	\end{equation*}
	\begin{equation*}
		\begin{tikzpicture}
	\begin{pgfonlayer}{nodelayer}
		\node [style=none] (6) at (-0.75, 0.75) {};
		\node [style=none] (7) at (-0.75, -0.25) {};
		\node [style=none] (12) at (0.75, -0.25) {};
		\node [style=none] (14) at (0.75, 0.75) {};
		\node [style=none] (15) at (-0.75, 0.25) {};
		\node [style=none] (16) at (-0.75, -0.75) {};
		\node [style=none] (17) at (0.75, -0.75) {};
		\node [style=none] (18) at (0.75, 0.25) {};
		\node [style=Z phase dot] (19) at (-1.5, 0.5) {};
		\node [style=none] (23) at (-2.25, 0.75) {};
		\node [style=none] (25) at (-2.25, 0.25) {};
		\node [style=none] (26) at (-0.75, -0.25) {};
		\node [style=none] (27) at (-0.75, -0.75) {};
		\node [style=Z phase dot] (28) at (-1.5, -0.5) {};
		\node [style=none] (29) at (-2.25, -0.25) {};
		\node [style=none] (30) at (-2.25, -0.75) {};
		\node [style=none] (31) at (2.25, 0.75) {};
		\node [style=none] (32) at (2.25, 0.25) {};
		\node [style=Z phase dot] (33) at (1.5, 0.5) {};
		\node [style=none] (34) at (0.75, 0.75) {};
		\node [style=none] (35) at (0.75, 0.25) {};
		\node [style=none] (36) at (2.25, -0.25) {};
		\node [style=none] (37) at (2.25, -0.75) {};
		\node [style=Z phase dot] (38) at (1.5, -0.5) {};
		\node [style=none] (39) at (0.75, -0.25) {};
		\node [style=none] (40) at (0.75, -0.75) {};
	\end{pgfonlayer}
	\begin{pgfonlayer}{edgelayer}
		\draw [in=180, out=0, looseness=0.75] (6.center) to (12.center);
		\draw [in=-180, out=0] (7.center) to (14.center);
		\draw [in=180, out=0, looseness=0.75] (15.center) to (17.center);
		\draw [in=-180, out=0] (16.center) to (18.center);
		\draw [style=none, in=345, out=180, looseness=1.25] (15.center) to (19);
		\draw [style=none, in=15, out=180, looseness=1.25] (6.center) to (19);
		\draw [style=none, in=0, out=165] (19) to (23.center);
		\draw [style=none, in=0, out=-165] (19) to (25.center);
		\draw [style=none, in=345, out=180, looseness=1.25] (27.center) to (28);
		\draw [style=none, in=15, out=180, looseness=1.25] (26.center) to (28);
		\draw [style=none, in=0, out=165] (28) to (29.center);
		\draw [style=none, in=0, out=-165] (28) to (30.center);
		\draw [style=none, in=345, out=180, looseness=1.25] (32.center) to (33);
		\draw [style=none, in=15, out=180, looseness=1.25] (31.center) to (33);
		\draw [style=none, in=0, out=165] (33) to (34.center);
		\draw [style=none, in=0, out=-165] (33) to (35.center);
		\draw [style=none, in=345, out=180, looseness=1.25] (37.center) to (38);
		\draw [style=none, in=15, out=180, looseness=1.25] (36.center) to (38);
		\draw [style=none, in=0, out=165] (38) to (39.center);
		\draw [style=none, in=0, out=-165] (38) to (40.center);
	\end{pgfonlayer}
\end{tikzpicture}

	\end{equation*}
\end{lemma}

\begin{proof}
	Using the definition of the functor $Z$ and by applying the spider fusion law of the ZX-calculus.
\end{proof}

\begin{definition}
	\label{def:K}
	Let $K: \catname{ZX}_{\text{doubled}} \rightarrow \catname{ZX}$ be the functor which acts as division by 2 on the objects and on morphisms it is defined by a map $k$:
	\vspace*{-1ex}
	\begin{equation*}
		 \xmapsto{k} \begin{tikzpicture}
	\begin{pgfonlayer}{nodelayer}
		\node [style=X dot] (0) at (0, 0) {$l \pi$};
		\node [style=none] (1) at (-1, 1) {};
		\node [style=none] (3) at (1, 1) {};
		\node [style=none] (6) at (1, -1) {};
		\node [style=none] (8) at (-1, -1) {};
		\node [style=none] (17) at (-1.25, 0) {$\vdots$};
		\node [style=none] (18) at (1.25, 0) {$\vdots$};
		\node [style=none] (19) at (-1.75, 0) {$m$};
		\node [style=none] (20) at (1.75, 0) {$n$};
	\end{pgfonlayer}
	\begin{pgfonlayer}{edgelayer}
		\draw [in=165, out=0] (1.center) to (0);
		\draw [in=-180, out=15] (0) to (3.center);
		\draw [in=0, out=-165] (0) to (8.center);
		\draw [in=180, out=-15] (0) to (6.center);
	\end{pgfonlayer}
\end{tikzpicture}
, \qquad
		 \xmapsto{k} \begin{tikzpicture}
	\begin{pgfonlayer}{nodelayer}
		\node [style=Z dot] (0) at (0, 0) {$l \pi$};
		\node [style=none] (1) at (-1, 1) {};
		\node [style=none] (3) at (1, 1) {};
		\node [style=none] (6) at (1, -1) {};
		\node [style=none] (8) at (-1, -1) {};
		\node [style=none] (17) at (-1.25, 0) {$\vdots$};
		\node [style=none] (18) at (1.25, 0) {$\vdots$};
		\node [style=none] (19) at (-1.75, 0) {$m$};
		\node [style=none] (20) at (1.75, 0) {$n$};
	\end{pgfonlayer}
	\begin{pgfonlayer}{edgelayer}
		\draw [in=165, out=0] (1.center) to (0);
		\draw [in=-180, out=15] (0) to (3.center);
		\draw [in=0, out=-165] (0) to (8.center);
		\draw [in=180, out=-15] (0) to (6.center);
	\end{pgfonlayer}
\end{tikzpicture}
,
	\end{equation*}
	\begin{equation*}
		 \xmapsto{k} \begin{tikzpicture}
	\begin{pgfonlayer}{nodelayer}
		\node [style=none] (1) at (-0.75, -0.5) {};
		\node [style=none] (2) at (0.75, 0.5) {};
		\node [style=none] (3) at (0.75, -0.5) {};
		\node [style=none] (4) at (-0.75, 0.5) {};
	\end{pgfonlayer}
	\begin{pgfonlayer}{edgelayer}
		\draw [in=180, out=0] (1.center) to (2.center);
		\draw [in=0, out=180] (3.center) to (4.center);
	\end{pgfonlayer}
\end{tikzpicture}

	\end{equation*}
	The identity, cups, and caps can be obtained from a 2-legged spider with appropriate domain/codomain.
    \vspace*{-1ex}
\end{definition}

The functor $K$ can be seen as a projection of a 2-qubit space, where the $\ket{+}$ is encoded by a Bell state, to a 1-qubit space. The ``physical'' maps acting on the 2-qubit space are each interpreted as maps acting on the ``logical'' 1-qubit space. We are simply using a 4-dimensional space to represent a qubit. This is the difference between the functor $K \circ Z$ (depicted in Table~\ref{table:KNOT-doubled-ZX-semantics}) and the functor $Z$ -- in effect a further level of encoding in which logical qubits are represented by (the correlation surfaces on) pairs of defects rather than single defects.

\begin{table}[t]
    \centering
    \vspace*{-3ex}
\noindent
	\begin{tabular}{c}
		\begin{tabular}{cc}
			$\input{zx-knot/spider1-r-def.tikz} \xmapsto{k \circ z} \input{zx/doubled/undoubled/red-spider.tikz}$
			&
			$\input{zx-knot/spider2-r-def.tikz} \xmapsto{k \circ z} \input{zx/doubled/undoubled/green-spider.tikz}$
		\end{tabular}\\
		\begin{tabular}{c}
			$\input{knot-doubled/swap.tikz}\xmapsto{k \circ z} \input{zx/doubled/undoubled/swap.tikz}$
		\end{tabular}
	\end{tabular}
	~
    \caption{%
        \label{table:KNOT-doubled-ZX-semantics}%
        We modify the semantics specified by the functor $Z$ to account for the doubled structure of $\catname{KNOT}_{\text{doubled}}$. By binary $l$ we done the presence or an absence of an operator. We make use of the functor $K$ (acting on maps as $k$) to lift the semantics of $\catname{KNOT}$ to $\catname{KNOT}_{\text{doubled}}$.}
    \vspace*{-1ex}
    \end{table}

\begin{proposition}
\label{prop:zx-analogue-knot}
We have the following correspondence between $\catname{KNOT}_{\text{doubled}}$ and ($0$, $\pi$)-ZX-calculus: \ref{lemma:green-spider-zx-knot}, \ref{lemma:red-spider-zx-knot}, \ref{lemma:red-spider-zx-knot-state} and \ref{lemma:green-spider-zx-knot-state} are the spider rules for both colours, eq. \ref{lemma:red-spider-copy-thr-green-zx-knot} and \ref{lemma:green-spider-copy-thr-red-zx-knot} are the copy rules for both colours, eq. \ref{lemma:bialgebra} is the bialgebra law, eq. \ref{lemma:bar-copies}, \ref{lemma:circle-passes}, \ref{lemma:bar-passes} and \ref{lemma:circle-copies} are the $\pi$-copy rules.
\end{proposition}

\begin{proof}
Convert to ZX-calculus using the functor $Z$ and apply the spider fusion rule.
\end{proof}

\begin{theorem}
	$\catname{KNOT}_{\text{doubled}}$ is sound and complete for the ($0$, $\pi$)-fragment of ZX-calculus.

	\hfill
	\textup{[Proof in Appendix \ref{appendix-prove-soundness}].}
\end{theorem}

\vspace{-0.7cm}
\section{Conclusions}
	In this work, we introduced $\catname{KNOT}$, a graphical calculus for defect-based surface code computing. By leveraging this formalism, we provided a systematic framework for reasoning about defect braiding, bridging the gap between informal graphical approaches and algebraic descriptions. Furthermore, we extended this framework via a doubling construction that captures conventional encoding techniques in defect-based quantum computation, and demonstrated its equivalence to the ($0$, $\pi$)-fragment of ZX-calculus. We established soundness and completeness results, showing that both the tangle-like and ribbon-like sublanguages of $\catname{KNOT}$ faithfully represent ($0$, $\pi$)-ZX-diagrams. This formalisation enables automated verification and optimisation of defect braiding protocols, facilitating a deeper understanding of topological quantum computation. This also provides a formal way of understanding the semantics of braided surface codes.

	We currently have partial results demonstrating how this framework can be used to reason about correlation surfaces in surface codes. In particular, we are working towards a more general diagrammatic calculus capturing the equivalences of correlation surfaces. The usual boundary conditions (an infinite lattice) yield a structure of a Hopf-Frobenius algebra where the correlation surface between a pair of defects satisfies the relations of a fragment of ZX-calculus.

	Future work includes extending $\catname{KNOT}$ to incorporate byproduct operations and generalising its applicability beyond defect braiding with specific boundary conditions to broader classes of topological codes. The study of the connections between the homology of surface codes and different graphical calculi is also a promising direction for future research. Working towards capuring the homological equivalences for surfaces code for arbitrary topologies and boundaries using diagrammatics could help to build tools for compilation and verification for a broader class of surfaces codes. We also plan to explore the use of $\catname{KNOT}$ for optimising defect braiding protocols in practical quantum computing scenarios, including the development of automated tools for verifying and optimising defect-based quantum circuits.

\paragraph{Acknowledgements.} We would like to thank Robert Booth for the feedback on an earlier draft of this paper. We would also like to thank Alexandru Paler for useful discussions on braided surface codes.

\nocite{*}
\bibliographystyle{eptcs}
\bibliography{generic}

\newpage
\appendix

\section{ZX-calculus}
Graphical languages can be formalised as categories. One formulation is to treat the diagrams as morphisms and input/outputs or strings as objects (since diagrams represent processes which are modelled by morphisms) (e.g. see Ref. \cite{cockett2017category, comfort2019circuit, bonchi2017interacting}). Other examples include formulating diagrams as a bicategory \cite{cicala2017categorifying}, algebraic \cite{kissinger2012pictures, penrose1971applications}, topological \cite{kissinger2012pictures, joyal1991geometry}, and combinatoric \cite{dixon2008extending} formulations.
In particular, string diagrams are often represented as a $\catname{PROP}$ \cite{wang2018completeness}.

\begin{definition}
	A $\catname{PROP}$ is a strict symmetric monoidal category where every object is of the form
	\begin{equation}
		x^{\bigotimes n} = x \otimes x \otimes \dots \otimes x
	\end{equation}
	for a single object $x$ and $n \geq 0$.
\end{definition}

One can also consider a $\catname{PROP}$ as a category of natural numbers with addition as a tensor product. Since many graphical languages can be considered as a $\catname{PROP}$, we use it to define maps between our languages as categories.

In considering the graphical languages, we follow the approach of \cite{wang2018completeness}. For any diagrammatic system in this article, we have $S$ the set of basic diagrams together with an empty diagram. As in this chapter we interpret the diagrams from left (inputs) to right (outputs), we can compose the basic diagrams horizontally ($\circ$) and vertically ($\otimes$). The diagrams satisfy a set of equivalence relations $R$. Thus, any category $\catname{C}$ having diagrams as morphisms based on a generator set $S$ and satisfying the relations $R$ is implicitly considered modulo $R$ ($\catname{C}/R$) in this chapter. Any two morphisms are the same in $\catname{C}/R$ if their diagrammatic representations in $\catname{C}$ are equal according to $R$. This has a role in considering the mappings between different categories based on graphical languages. A mapping $F: \catname{A}/R_A \rightarrow \catname{B}/R_B$ between two categories $\catname{A}$ modulo relations $R_A$ and $\catname{B}$ modulo relations $R_B$ is functorial if we map the generators of $\catname{A}$ to $\catname{B}$ and if the relations $R_A$ of $\catname{A}$ hold in $\catname{B}$ after the mapping $F$. This ensures that two diagrams equal in $\catname{A}$ (representing the same morphism in $\catname{A}/R_A$) are mapped to one morphism in $\catname{B}/R_B$ (that the functor $F$ is well-defined).

\subsection{ZX-calculus}
In this work, we focus on a specific fragment of ZX-calculus where all the phases are either $0$ and $\pi$. This is a subcategory of ZX-calculus similar to the calculi presented in Refs.~\cite{kissinger2022phasefree}. This does not preserve some of the semantics of ZX-calculus in terms of $\catname{FHilb}$ (we do not have any rules involving zero scalars), but the rewrite rules are sufficient for tasks such as ZX-diagram simplification \cite{pyzx}, circuit optimisation \cite{Duncan2020graphtheoretic}, stabiliser simulation \cite{Kissinger-2022}. Therefore,  soundness and completeness we give is in terms of the ``most common'' rules of ZX-calculus.

In ZX-calculus, there are two types of spiders: \textit{Z spiders} and \textit{X spiders}. These are defined with respect to the eigenbases of the Pauli Z and Pauli X operators, respectively.
\begin{equation}\label{eq:spiders}
	\begin{array}{ccl}
		\textrm{\scriptsize $m$}\!\left\{\ \ \input{zx/z-phase-spider.tikz}\ \ \right\}\!\textrm{\scriptsize $n$}
		 & \ := \   &
		\ket{0}^{\otimes n}\bra{0}^{\otimes m} + e^{i\alpha}\ket{1}^{\otimes n}\bra{1}^{\otimes m} \\[5mm]
		\textrm{\scriptsize $m$}\!\left\{\ \ \input{zx/x-phase-spider.tikz}\ \ \right\}\!\textrm{\scriptsize $n$}
		 & \ := \   &
		\ket{{+}}^{\otimes n}\bra{{+}}^{\otimes m} +
		e^{i\alpha}\ket{{-}}^{\otimes n}\bra{{-}}^{\otimes m}
	\end{array}
\end{equation}

We usually omit phases if $\alpha=0$ inside spiders as a convention.

In addition to spiders, we allow identity wires, swaps, cups, and caps in ZX diagrams, which are defined as follows:
\[
	\begin{tikzpicture}
	\begin{pgfonlayer}{nodelayer}
		\node [style=none] (12) at (0, -0.75) {};
		\node [style=none] (14) at (0, 0.75) {};
		\node [style=none] (15) at (0, 0.75) {};
		\node [style=none] (16) at (0, -0.75) {};
	\end{pgfonlayer}
	\begin{pgfonlayer}{edgelayer}
		\draw [bend left=90, looseness=1.75] (16.center) to (15.center);
	\end{pgfonlayer}
\end{tikzpicture}
 \ \ := \ \ \sum_i \ket{ii} \qquad
	\begin{tikzpicture}
	\begin{pgfonlayer}{nodelayer}
		\node [style=none] (12) at (0, -0.75) {};
		\node [style=none] (14) at (0, 0.75) {};
		\node [style=none] (15) at (0, 0.75) {};
		\node [style=none] (16) at (0, -0.75) {};
	\end{pgfonlayer}
	\begin{pgfonlayer}{edgelayer}
		\draw [bend right=90, looseness=1.75] (16.center) to (15.center);
	\end{pgfonlayer}
\end{tikzpicture}
 \ \ := \ \ \sum_i \bra{ii} \qquad
	\tikz[tikzfig]{ \draw (0,0) -- (3,0); } \ \ :=\ \  \sum_i \ket{i}\bra{i}\qquad
	\input{zx/swap.tikz} \ \ :=\ \ \sum_{ij} \ket{ij}\bra{ji}
\]
Finally, we have the following rewrite rules for the ($0$, $\pi$)-fragment of ZX-calculus modulo scalars \cite{kissinger2022phasefree}:
\begin{center}
	\begin{tikzpicture}
	\begin{pgfonlayer}{nodelayer}
		\node [style=Z phase dot] (1) at (-10, 1.25) {$k \pi$};
		\node [style=none] (2) at (-11.5, 2) {};
		\node [style=none] (3) at (-11.5, 0.5) {};
		\node [style=none] (4) at (-8, 2) {};
		\node [style=none] (5) at (-8, 0.5) {};
		\node [style=none] (6) at (-7.75, 1.5) {$\vdots$};
		\node [style=none] (7) at (-11.25, 1.5) {$\vdots$};
		\node [style=Z phase dot] (8) at (-9, -1.75) {$l \pi$};
		\node [style=none] (9) at (-11, -1) {};
		\node [style=none] (10) at (-11, -2.5) {};
		\node [style=none] (11) at (-7.5, -1) {};
		\node [style=none] (12) at (-7.5, -2.5) {};
		\node [style=none] (13) at (-7.75, -1.5) {$\vdots$};
		\node [style=none] (14) at (-11.25, -1.5) {$\vdots$};
		\node [style=Z phase dot] (15) at (-3, -0.25) {$(l + k) \pi$};
		\node [style=none] (16) at (-1.5, 2) {};
		\node [style=none] (17) at (-1.5, 0.5) {};
		\node [style=none] (18) at (-1.5, -1) {};
		\node [style=none] (19) at (-1.5, -2.5) {};
		\node [style=none] (20) at (-4.5, 2) {};
		\node [style=none] (21) at (-4.5, 0.5) {};
		\node [style=none] (22) at (-4.5, -1) {};
		\node [style=none] (23) at (-4.5, -2.5) {};
		\node [style=none] (24) at (-1.75, 1.5) {$\vdots$};
		\node [style=none] (25) at (-1.75, -1.5) {$\vdots$};
		\node [style=none] (26) at (-4.25, 1.5) {$\vdots$};
		\node [style=none] (27) at (-4.25, -1.5) {$\vdots$};
		\node [style=none] (28) at (-11.5, -1) {};
		\node [style=none] (29) at (-11.5, -2.5) {};
		\node [style=none] (30) at (-7.5, 2) {};
		\node [style=none] (31) at (-7.5, 0.5) {};
		\node [style=none] (32) at (-6, -0.25) {$=$};
		\node [style=X phase dot] (67) at (2.75, 1.5) {$\pi$};
		\node [style=Z phase dot] (68) at (4, 1.5) {$k \pi$};
		\node [style=none] (69) at (7, 1.5) {$=$};
		\node [style=none] (70) at (5, 2.25) {};
		\node [style=none] (71) at (5, 0.75) {};
		\node [style=none] (72) at (2, 1.5) {};
		\node [style=Z phase dot] (73) at (9.75, 1.5) {$k \pi$};
		\node [style=X phase dot] (74) at (11, 2.25) {$\pi$};
		\node [style=X phase dot] (75) at (11, 0.75) {$\pi$};
		\node [style=none] (76) at (8.25, 1.5) {};
		\node [style=none] (77) at (11.75, 2.25) {};
		\node [style=none] (78) at (11.75, 0.75) {};
		\node [style=X phase dot] (79) at (-11.5, -4.75) {};
		\node [style=Z phase dot] (80) at (-10.5, -4.75) {};
		\node [style=none] (81) at (-8, -4.75) {$=$};
		\node [style=X phase dot] (82) at (-6.5, -4) {};
		\node [style=X phase dot] (83) at (-6.5, -5.5) {};
		\node [style=none] (84) at (-5.5, -4) {};
		\node [style=none] (85) at (-5.5, -5.5) {};
		\node [style=none] (86) at (-9.75, -4) {};
		\node [style=none] (87) at (-9.75, -5.5) {};
		\node [style=none] (98) at (7, -1.5) {$=$};
		\node [style=none] (159) at (8.75, -0.875) {};
		\node [style=none] (160) at (8.75, -2.125) {};
		\node [style=X phase dot] (161) at (9.5, -1.5) {};
		\node [style=Z phase dot] (162) at (10.5, -1.5) {};
		\node [style=none] (163) at (11.25, -0.875) {};
		\node [style=none] (164) at (11.25, -2.125) {};
		\node [style=none] (165) at (2.75, -0.875) {};
		\node [style=none] (166) at (2.75, -2.125) {};
		\node [style=Z phase dot] (167) at (3.375, -0.875) {};
		\node [style=Z phase dot] (168) at (3.375, -2.125) {};
		\node [style=X phase dot] (169) at (4.875, -0.875) {};
		\node [style=X phase dot] (170) at (4.875, -2.125) {};
		\node [style=none] (171) at (5.5, -0.875) {};
		\node [style=none] (172) at (5.5, -2.125) {};
		\node [style=none] (173) at (8.5, -0.875) {};
		\node [style=none] (174) at (8.5, -2.125) {};
		\node [style=none] (175) at (11.5, -0.875) {};
		\node [style=none] (176) at (11.5, -2.125) {};
		\node [style=none] (177) at (5.5, 0.75) {};
		\node [style=none] (178) at (5.5, 2.25) {};
		\node [style=none] (179) at (-9.5, -4) {};
		\node [style=none] (180) at (-9.5, -5.5) {};
		\node [style=X phase dot] (213) at (-2.25, -4.75) {};
		\node [style=Z phase dot] (214) at (-1.25, -4.75) {};
		\node [style=none] (215) at (0, -4.5) {$=$};
		\node [style=none] (216) at (1.25, -3.875) {};
		\node [style=none] (217) at (1.25, -5.125) {};
		\node [style=none] (218) at (2.5, -3.875) {};
		\node [style=none] (219) at (2.5, -5.125) {};
		\node [style=none] (220) at (-6, 0.25) {(Spider)};
		\node [style=none] (221) at (7, 2) {$\pi$};
		\node [style=none] (222) at (7, -1) {(Bialgebra)};
		\node [style=none] (223) at (0, -4) {(One)};
		\node [style=none] (224) at (-8, -4.25) {(Copy)};
	\end{pgfonlayer}
	\begin{pgfonlayer}{edgelayer}
		\draw [in=180, out=30] (1) to (4.center);
		\draw [in=180, out=-30] (1) to (5.center);
		\draw [in=360, out=120] (1) to (2.center);
		\draw [in=0, out=-120] (1) to (3.center);
		\draw [in=180, out=60] (8) to (11.center);
		\draw [in=180, out=-60] (8) to (12.center);
		\draw [in=360, out=150] (8) to (9.center);
		\draw [in=0, out=-150] (8) to (10.center);
		\draw [in=180, out=30] (15) to (17.center);
		\draw [in=180, out=60] (15) to (16.center);
		\draw [in=180, out=-30] (15) to (18.center);
		\draw [in=180, out=-60] (15) to (19.center);
		\draw [in=0, out=120] (15) to (20.center);
		\draw [in=0, out=150] (15) to (21.center);
		\draw [in=0, out=-150] (15) to (22.center);
		\draw [in=0, out=-120] (15) to (23.center);
		\draw (28.center) to (9.center);
		\draw (29.center) to (10.center);
		\draw (5.center) to (31.center);
		\draw (4.center) to (30.center);
		\draw [in=135, out=-45] (1) to (8);
		\draw (68) to (67);
		\draw [in=180, out=60] (68) to (70.center);
		\draw [in=180, out=-60] (68) to (71.center);
		\draw (72.center) to (67);
		\draw [in=180, out=60] (73) to (74);
		\draw [in=180, out=-60] (73) to (75);
		\draw (75) to (78.center);
		\draw (77.center) to (74);
		\draw (76.center) to (73);
		\draw (80) to (79);
		\draw (83) to (85.center);
		\draw (82) to (84.center);
		\draw [in=180, out=60] (80) to (86.center);
		\draw [in=180, out=-60] (80) to (87.center);
		\draw (161) to (162);
		\draw [in=0, out=-105, looseness=0.75] (161) to (160.center);
		\draw [in=0, out=105, looseness=0.75] (161) to (159.center);
		\draw [in=180, out=75, looseness=0.75] (162) to (163.center);
		\draw [in=-180, out=-75, looseness=0.75] (162) to (164.center);
		\draw (165.center) to (167);
		\draw (166.center) to (168);
		\draw (169) to (171.center);
		\draw (170) to (172.center);
		\draw (167) to (170);
		\draw (168) to (169);
		\draw [bend left=15] (167) to (169);
		\draw [bend right=15] (168) to (170);
		\draw (160.center) to (174.center);
		\draw (173.center) to (159.center);
		\draw (175.center) to (163.center);
		\draw (176.center) to (164.center);
		\draw (177.center) to (71.center);
		\draw (178.center) to (70.center);
		\draw (180.center) to (87.center);
		\draw (179.center) to (86.center);
		\draw (214) to (213);
		\draw [dotted] (216.center)
			 to (218.center)
			 to (219.center)
			 to (217.center)
			 to cycle;
	\end{pgfonlayer}
\end{tikzpicture}
,
\end{center}

where $l, k$ are in $\mathbb{F}_2$.

\paragraph{Note.}

The rules given above are non-standard: we ignore scalars completely and only allow the phases to be either 0 or $\pi$ (which introduces the $X$, $Z$ commutation rule). Strictly speaking, this category is actually the subcategory of the Spekkens' toy theory ($\catname{Spek}_2$) \cite{spekkens2007evidence} where the phase group is $\mathbb{Z}_2$ (if we only allow the spiders to be decorated with $11$ as in the convention from \cite{backens2016complete} or alternatively with $\pi, \pi$ following \cite[\S 11.2.2]{Coecke-Kissinger-2017})  -- we will describe this relationship more in section~\ref{section:proof-soundness}. Interestingly, even though this is a subcategory of $\catname{Spek}_2$, because we do not have the problematic $\pm \frac{\pi}{2}$ ($01$, $10$ from Ref. \cite{backens2016complete}) phases, this is also a subcategory of $\catname{Stab}_2$ -- we can define a functor from this category into $\catname{Stab}_2$.

\begin{proof}[Proof of Proposition~\ref{prop:z-functor-zx}]
	\label{proof:z-functor-zx}
	We need to show that the rules of $\catname{KNOT}$ hold in ZX-calculus via the interpretation map $z$. We only explicitly prove the rules which require multiple steps. All other rules of $\catname{KNOT}$ can be proven to hold in ZX-calculus with the application of a single ZX-calculus rule.
	\begin{itemize}
		\item Rule \ref{eq:loop-around-one-wire}:
		      \begin{equation*}
			      \input{knot-generators-tikzit/K8-L.tikz} \mapsto \input{zx/loop-around.tikz} = \input{zx/loop-around-1.tikz} = \input{zx/loop-around-2.tikz} \mapsfrom \input{knot-generators-tikzit/K8-R.tikz}
		      \end{equation*}
		\item Rule \ref{eq:K6}:
		      \begin{equation*}
			      \input{knot-generators-tikzit/K5-L.tikz} \mapsto \input{zx/z-cnot-twice.tikz} = \input{zx/z-cnot-twice-1.tikz} = \input{zx/z-cnot-twice-2.tikz} = \begin{tikzpicture}
	\begin{pgfonlayer}{nodelayer}
		\node [style=none] (2) at (1, -0.75) {};
		\node [style=none] (3) at (1, 0.75) {};
		\node [style=none] (4) at (-1, -0.75) {};
		\node [style=none] (5) at (-1, 0.75) {};
	\end{pgfonlayer}
	\begin{pgfonlayer}{edgelayer}
		\draw (5.center) to (3.center);
		\draw (4.center) to (2.center);
	\end{pgfonlayer}
\end{tikzpicture}
 \mapsfrom \input{knot-generators-tikzit/K5-R.tikz}
		      \end{equation*}
		\item Rule \ref{eq:overcross-kill}:
			\begin{equation*}
				\input{knot-generators-tikzit/K21-L.tikz} \mapsto \input{zx/dot-kills-overcrossing.tikz} = \input{zx/dot-kills-overcrossing-1.tikz} = \input{zx/dot-kills-overcrossing-2.tikz} \mapsfrom \input{knot-generators-tikzit/K21-R.tikz}
			\end{equation*}
		\item Rule \ref{eq:K7}:
				\begin{equation*}
					\input{knot-generators-tikzit/K9-L.tikz} \mapsto \input{zx/decorators-copy.tikz} = \input{zx/decorators-copy-1.tikz} = \input{zx/decorators-copy-2.tikz} \mapsfrom \input{knot-generators-tikzit/K9-R.tikz}
				\end{equation*}
		\item Rule \ref{eq:loop-D-commute}:
			the case for $p=0$ and $q=0$ is trivial. If we translate the subdiagram $D$ into ZX-calculus, the resulting diagram will be constructed from the following generators:
			\begin{equation*}
				\input{zx/z-cnot.tikz} \quad \input{zx/x-cnot.tikz} \quad \input{zx/green-merge.tikz} \quad
				\input{zx/red-merge.tikz}.
			\end{equation*}
			The loop translated in ZX-calculus will have the following two forms depending on $p$ and $q$:
			\begin{equation*} \tag{*}
				\label{eq:loops-in-zx}
				\input{zx/loop-around-a.tikz} \quad \input{zx/two-loops-around.tikz}.
			\end{equation*}

			All $\input{zx/z-cnot.tikz}$ and $\input{zx/x-cnot.tikz}$ will be positioned such that the loops from \eqref{eq:loops-in-zx} will commute through them (\gatename{CNOT}s are always positioned such that we can only get green spider on a wire corresponding to a black defect and red spider on a wire corresponding to a grey defect). For the spiders we show the following:
			\begin{equation*}
				\input{zx/D-commute-proof.tikz} = 				\input{zx/D-commute-proof-1.tikz} = 				\input{zx/D-commute-proof-2.tikz}  = \dots = 				\input{zx/D-commute-proof-5.tikz}.
			\end{equation*}
			On the left hand side of the red spider, we might end up with green spiders (if the number of input wires is not one). We can however always remove the connections, since the wire might have a red spider connected to it (which would kill the connection) or (in $\catname{KNOT}$) the equivalent of the loop might be connected to the wire. In this case applying the equation above would kill the connection by the Hopf rule. In the end, we are able to commute the equivalent of the loop through the diagram $z(D)$.
	\end{itemize}
\end{proof}

\subsection{Soundness and completeness of \texorpdfstring{$\catname{KNOT}_{\text{doubled}}$}{KNOT_{doubled}} for the ($0$, $\pi$)-fragment of ZX-calculus}
\label{appendix-prove-soundness}

\begin{lemma}
	\label{lemma:nf-braided}
	The diagrams of the category $\catname{KNOT}_{\text{doubled}}$ have a normal form.
\end{lemma}

\begin{proof}
To prove this we will use the availability of a normal form for ($0$, $\pi$)-fragment of ZX-calculus (modulo scalars).

First, for any given $\catname{KNOT}_{\text{doubled}}$ diagram, let's take a ribbon-like structure (LHS of eq.~\eqref{duality-1} and eq.~\eqref{duality-2}) and rewrite it to the corresponding tangle-like structure (RHS) for which we provided explicit definitions of $\catname{ZX}_{\text{doubled}}$ (Definition~\ref{lemma:braided-zx-gens}) and of the functor $K$ (Definition~\ref{def:K}).

Notice that the functor $K$ is full and faithful. If we restrict the functor $Z$ to $\catname{KNOT}_{\text{doubled}}$, the restricted functor $\widetilde{Z}$ is also full and faithful. This makes the composite $K \circ \widetilde{Z}$ a full and faithful functor from $\catname{KNOT}_{\text{doubled}}$ to $\catname{ZX}$.

($0$, $\pi$)-fragment of ZX-calculus (modulo scalars) has a unique normal form. Let us consider two $\catname{KNOT}_{\text{doubled}}$ diagrams $d_1$ and $d_2$. If $K \circ \widetilde{Z}(d_1)$ and $K \circ \widetilde{Z}(d_2)$ reduce to the same normal form, then the $d_1$ and $d_2$ reduce to a corresponding $\catname{KNOT}_{\text{doubled}}$ diagram $d_r$ (also by Proposition~\ref{prop:zx-analogue-knot} $\catname{KNOT}_{\text{doubled}}$ is sound for ($0$, $\pi$)-fragment of ZX-calculus). We always stay in the image of the functor $K \circ \widetilde{Z}$.

In the other direction, if $d_1$ and $d_2$ are equal in $\catname{KNOT}_{\text{doubled}}$, then they are equal in $\catname{ZX}_{\text{doubled}}$ ($Z(d_1) = Z(d_2)$) using the soundness result of Proposition~\ref{prop:z-functor-zx}. This means that we also have $K \circ Z(d_1) = K \circ Z(d_2)$, which means that the ZX diagrams must have the same normal form, which by the fullness and faithfulness of the functor $K \circ \widetilde{Z}$ has a corresponding $\catname{KNOT}_{\text{doubled}}$ diagram $d_r$.
\end{proof}

Another way of seeing how this holds is by observing that all the rules of the ($0$, $\pi$)-fragment of ZX-calculus (modulo scalars) have an equivalent rule in $\catname{KNOT}_{\text{doubled}}$ by Proposition~\ref{prop:zx-analogue-knot} ($\catname{KNOT}_{\text{doubled}}$ is sound for the ($0$, $\pi$)-fragment of ZX-calculus (modulo scalars)). ($0$, $\pi$)-fragment of ZX-calculus (modulo scalars) has a unique normal form which is obtained by using a rewrite system composed of a spider, bialgebra, and Hopf rule (which is provable from other rules). We can therefore use the same rewrite system in $\catname{KNOT}_{\text{doubled}}$ to obtain a $\catname{KNOT}_{\text{doubled}}$ diagram. At each stage of the rewrite procedure, the $\catname{KNOT}$ diagrams will also be $\catname{KNOT}_{\text{doubled}}$ diagrams. The resulting diagram will be a normal form for $\catname{KNOT}_{\text{doubled}}$ diagrams only with respect to the rules from Proposition~\ref{prop:zx-analogue-knot}. As all the rules in $\catname{KNOT}_{\text{doubled}}$ come from the rules of $\catname{KNOT}$, we need to show that the normal form is a normal form with respect to the remaining rules of $\catname{KNOT}$.

Suppose that we have two $\catname{KNOT}_{\text{doubled}}$ diagrams whose above normal forms are different (by the rules of Proposition~\ref{prop:zx-analogue-knot}). If we map these two normal forms into $\catname{ZX}_{\text{doubled}}$ via the functor $Z$, the resulting two diagrams must be different, but from the definition of the functor $Z$, the set of all the rules of $\catname{KNOT}$ is a subset of the set of all the equivalent rules of $\catname{ZX}_{\text{doubled}}$. Therefore, the two diagrams must also be different in $\catname{KNOT}_{\text{doubled}}$ according to the rules of $\catname{KNOT}$.

\begin{lemma}
	\label{lemma:tangle-only}
	All identities holding in $\catname{KNOT}_{\text{doubled}}$ via the rules of $\catname{KNOT}$ are derivable from the relations in Proposition~\ref{prop:zx-analogue-knot}.
\end{lemma}

\begin{proof}
This is a corollary of Lemma~\ref{lemma:nf-braided}. If two $\catname{KNOT}_{\text{doubled}}$ diagrams are equal according to the rules $\catname{KNOT}$, they have the same unique normal form which is obtained using the rules from Proposition~\ref{prop:zx-analogue-knot}.
\end{proof}

\begin{theorem}
	$\catname{KNOT}_{\text{doubled}}$ via the interpretation $K \circ Z$ is sound and complete for the ($0$, $\pi$)-fragment of ZX-calculus.
\end{theorem}

\begin{proof}
	Follows from the discussion above. Proposition~\ref{prop:zx-analogue-knot} shows that the rules of the ($0$, $\pi$)-fragment of ZX-calculus hold in $\catname{KNOT}_{\text{doubled}}$ and the Lemma~\ref{lemma:tangle-only} shows that they are the only rules.
\end{proof}

\section{Soundness of $\catname{KNOT}$ for the evolution of logical operators}
\label{section:proof-soundness}

\subsection{Overview}
The maps implemented by surface code braiding procedures are defined by their space-time stabilisers. These stabilisers are defined by how the logical (undetectable) errors evolve with the computation. The logical operators in the surface code are described by first homology groups of the underlying topology of a surface code lattice. To deduce a map implemented by a procedure we must thus have a way of capturing how the first homology groups evolve alongside the computation. Our assumptions about the lattice enable doing this with (affine) Lagrangian relations: a logical map implemented by a braiding procedure is fully characterised (modulo probabilistic byproducts) by an (affine) Lagrangian relation. The soundness property shown below is thus phrased in terms of the evolution of logical operators performed by \emph{some} physical surface code procedures implementing a particular braiding pattern.
The following ``cartoon'' commutative diagram summarises this reasoning:

\[
\xymatrix@C=1cm@R=1cm{
  \catname{KNOT}
      \ar[r]^-{G}
      \ar@{^{(}->}[d]^{I}                     &
  \catname{AffLagRel}_{\mathbb{F}_2}
      \ar[r]^-{\sim}                       &
  \catname{Spek}_{2}    \\
  \substack{\text{Braiding patterns}\\\text{(with undetectable errors)}}
      \ar[ru]^{G_{\text{extended}}}                      & & \\ 
  {%
  \substack{\text{Physical QEC Procedure}\\\text{(up to byproducts)}}
  } \ar[u]^{F_{\text{implement}}} & &
}
\]
Here, $\catname{Spek}_2$ is the category of Spekkens' toy bit theory (for alternative presentations, see~\cite{backens2016complete} and \cite[\S 11.2.2]{Coecke-Kissinger-2017}). We give the definition of the functor $G$ in section~\ref{sec:soundness-afflagrel}. We do not claim that $\catname{KNOT}$ captures all possible braiding patterns and equivalences between them, and $I$ represents a functor into a bigger category of all such braiding patterns, containing $\catname{KNOT}$ is a subcategory. Affine Lagrangian relations over $\mathbb{F}_2$ are exactly the circuits belonging to the Spekkens' toy bit theory~\cite{backens2016complete, spekkens2007evidence}.

As we will see, the image of the functor $G$ does not have problematic phases causing a mismatch between $\catname{Spek}_2$ and $\catname{Stab}_2$ which is a subcategory of $\catname{Mat(\mathbb{C})}$ generated by the qubit Clifford group, $\bra{0}$, and $\ket{0}$, modulo invertible scalars (for a graphical presentation, see~\cite{simplified-stab-backens}). We will therefore treat the image of the functor $G$ as corresponding to morphisms of $\catname{Stab}_2$.

\subsection{Background preliminaries}
\subsubsection{Homological algebra}
\label{hom-alg}
We first introduce some background from homological algebra. In particular, the tool of choice is the chain complex.

\begin{definition}
	A chain complex $C_{\bullet}$ in the category $\catname{C}$ is: a collection of objects $\{C_n\}_{n \in \mathbb{Z}}$, and of morphisms
			      $\dots \xrightarrow[]{\partial_3} C_2 \xrightarrow[]{\partial_2} C_1 \xrightarrow[]{\partial_1} C_0 \xrightarrow[]{\partial_0} C_{-1}  \xrightarrow[]{\partial_{-1}} \dots$ ($\partial_n: C_n \rightarrow C_{n-1}$),
		      such that $\partial_n \circ \partial_{n+1} = 0$ for all $n \in \mathbb{Z}$.
\end{definition}

\begin{definition}
	Given a chain complex $C_{\bullet}$, let $Z_n(C_{\bullet}) = \mathrm{ker(\partial_{n-1})}$ and $B_n(C_{\bullet}) = \mathrm{im}(\partial_n)$ which are called respectively n-cycles and n-boundaries. A quotient $H_n(C_{\bullet}) = Z_n(C_{\bullet})/B_n(C_{\bullet})$ is the nth homology group of $C_{\bullet}$.
\end{definition}

A \textbf{homology group} provides an algebraic way to capture the structure of a chain complex by distinguishing between cycles and boundaries within it. For a given chain complex \( C_\bullet \), the \textbf{nth homology group} \( H_n(C_\bullet) \) is defined as the quotient group:

\[
	H_n(C_\bullet) = Z_n(C_\bullet) / B_n(C_\bullet),
\]

where \( Z_n(C_\bullet) = \ker(\partial_{n-1}) \) denotes the \textbf{n-cycles}, which are elements in \( C_n \) that map to zero in \( C_{n-1} \) under the boundary map \( \partial_{n-1} \), \( B_n(C_\bullet) = \operatorname{im}(\partial_n) \) denotes the \textbf{n-boundaries}, which are elements in \( C_n \) that arise as boundaries from \( C_{n+1} \) under the boundary map \( \partial_n \).

Thus, the homology group \( H_n(C_\bullet) \) identifies elements of \( Z_n(C_\bullet) \) (cycles) that are not boundaries, capturing information about "holes" or "voids" in the structure represented by the chain complex.

\begin{definition}
	\label{subchain-complex}
	A chain subcomplex $D_{\bullet}$ of a chain complex $C_{\bullet}$ is a chain complex with the collection of objects $D_n$ being the subobjects of $C_n$ such that the following resultant square commutes:

	\begin{center}
		$\xymatrix@R=1cm@C=1cm{
			\dots \ar[r]^{} & D_{n+1} \ar[r]^{\partial_{n}^{D_{\bullet}}} \ar@{^{(}->}[d]^{f_{n+1}} & D_{n} \ar[r]^{\partial_{n-1}^{D_{\bullet}}} \ar@{^{(}->}[d]^{f_{n}} & D_{n-1} \ar[r]^{} \ar@{^{(}->}[d]^{f_{n-1}} & \dots \\
			\dots \ar[r]^{} & C_{n+1} \ar[r]^{\partial_{n}^{C_{\bullet}}} & C_{n} \ar[r]^{\partial_{n-1}^{C_{\bullet}}} & C_{n-1} \ar[r]^{}  & \dots
			}$,
	\end{center}
	The subcomplex inclusion map induced by the monomorphisms is called the chain map in the category of chain complexes. We call it the chain inclusion map.
\end{definition}

\begin{definition}
	Let $C_{\bullet}$ be a chain complex and $D_{\bullet}$ its chain subcomplex. The quotient $C_{\bullet} / D_{\bullet}$ ($C_{n} / D_{n}$ for all $n$) is the chain complex of $D_{\bullet}$-relative chains. The chain homology of the quotient is the $D_{\bullet}$-relative homology of $C_{\bullet}$: $H_n(C_{\bullet}, D_{\bullet}) := H_n(C_{\bullet}/D_{\bullet})$.
\end{definition}

In relative homology, we study cycles in a chain complex \( C_\bullet \) that are allowed to ``end" in a subcomplex \( D_\bullet \). The \textbf{relative homology group} \( H_n(C_\bullet, D_\bullet) \) captures information about cycles in \( C_\bullet \) that are not boundaries within \( D_\bullet \), thus describing features of \( C_\bullet \) relative to \( D_\bullet \).

\begin{lemma}
	\label{lemma-slice-vector-space}
	Let \( H_n(C_\bullet) \) be the nth homology group of a chain complex \( C_\bullet \) over the field \( \mathbb{F}_2 \). Then \( H_n(C_\bullet) \) is isomorphic to the vector space \( \mathbb{F}_2^{\operatorname{rank}(H_n(C_\bullet))} \), where \( \operatorname{rank}(H_n(C_\bullet)) \) denotes the rank of the homology group, or the number of independent cycles in \( H_n(C_\bullet) \).

	Similarly, for a pair \( (C_\bullet, D_\bullet) \) where \( D_\bullet \) is a subcomplex of \( C_\bullet \), the nth relative homology group \( H_n(C_\bullet, D_\bullet) \) is isomorphic to \( \mathbb{F}_2^{\operatorname{rank}(H_n(C_\bullet, D_\bullet))} \).
\end{lemma}

\begin{definition}
	A cochain complex $C^\bullet$ in the category $\catname{C}$ is: a collection of objects $\{C^n\}_{n \in \mathbb{Z}}$, and of morphisms $
	\cdots \xrightarrow{d^{n-2}} C^{n-1} \xrightarrow{d^{n-1}} C^n \xrightarrow{d^n} C^{n+1} \xrightarrow{d^{n+1}} \cdots$, such that $d^{n+1} \circ d^n = 0$ for all $n \in \mathbb{Z}$.
\end{definition}

\begin{definition}
	Given a cochain complex $C^\bullet$, let $Z^n(C^\bullet) = \ker(d^n)$ and $B^n(C^\bullet) = \operatorname{im}(d^{n-1})$, which are called respectively $n$-cocycles and $n$-coboundaries. A quotient $H^n(C^\bullet) = Z^n(C^\bullet) / B^n(C^\bullet)$ is the $n$th cohomology group of $C^\bullet$.

	\medskip

	A cohomology group provides an algebraic way to capture the structure of a cochain complex by distinguishing between cocycles and coboundaries within it. For a given cochain complex $C^\bullet$, the $n$th cohomology group $H^n(C^\bullet)$ is defined as the quotient group:
	\[
	H^n(C^\bullet) = Z^n(C^\bullet) / B^n(C^\bullet),
	\]
	where $Z^n(C^\bullet) = \ker(d^n)$ denotes the $n$-cocycles, which are elements in $C^n$ that map to zero in $C^{n+1}$ under the coboundary map $d^n$, and $B^n(C^\bullet) = \operatorname{im}(d^{n-1})$ denotes the $n$-coboundaries, which are elements in $C^n$ that arise as coboundaries from $C^{n-1}$ under the coboundary map $d^{n-1}$.
	\end{definition}

\subsubsection{Affine Lagrangian relations}

The contents of this section are based on \cite{Comfort:2021axr}.

\begin{definition}
	The \( \dagger \)-compact closed $\catname{PROP}$ of linear relations over \( \mathbb{F}_2 \), \(\catname{LinRel}_{\mathbb{F}_2}\) is a category of natural numbers and linear relations (linear subspaces of \( \mathbb{F}_2^n \oplus \mathbb{F}_2^m \) for natural numbers $n, m$) with relational composition, a direct sum of linear subspaces as the tensor product, and the relational converse as the dagger.
\end{definition}

\begin{definition}
	Let $\catname{AffRel}_{\mathbb{F}_2}$ denote $\catname{PROP}$ whose morphisms $n \to m$ are the (possibly empty) affine subspaces of ${\mathbb{F}_2}^n \oplus {\mathbb{F}_2}^m$, with composition given by relational composition and tensor product given by the direct sum.
\end{definition}

\begin{definition}
	Given a field ${\mathbb{F}_2}$, the $\catname{PROP}$ of {\bf Lagrangian relations},  $\catname{LagRel}_{\mathbb{F}_2}$, has morphisms ${\mathbb{F}_2}^{2n}\to {\mathbb{F}_2}^{2m}$ as Lagrangian subspaces of the symplectic vector space ${\mathbb{F}_2}^{n+m} \oplus {\mathbb{F}_2}^{n+m}$ with symplectic form:
	\begin{equation}
		\omega :=
		\begin{bmatrix}
			0_n  & I_n \\
			I_n & 0_n
		\end{bmatrix}.
	\end{equation}
	Lagrangian subspaces $W$ are the subspaces, for which we have $W^{\omega} := \{ v \in V : \forall w \in W,\, \omega(v, w) = 0 \} = W$. Composition is given by relational composition and the tensor product is given by the direct sum.
\end{definition}

\begin{definition}
	Let $\catname{AffLagRel}_{\mathbb{F}_2}$ denote the $\catname{PROP}$ whose morphisms are the affine Lagrangian subspaces of ${\mathbb{F}_2}^{2n} \oplus {\mathbb{F}_2}^{2m}$ of some symplectic vector space ${\mathbb{F}_2}^{n+m} \oplus {\mathbb{F}_2}^{n+m}$.
\end{definition}

\begin{lemma}
	Given a field $\mathbb{F}_2$, the $\catname{PROP}$ of matrices over $\mathbb{F}_2$ under the direct sum, $\catname{LinRel}_{\mathbb{F}_2}$, is presented by the following generators and equations, for $a,b \in \mathbb{F}_2$:

	\begin{equation*}
		\input{frob/1.tikz} = \input{frob/2.tikz}, \hspace*{1cm} \input{frob/3.tikz} = \input{frob/4.tikz}, \hspace*{1cm} \input{frob/5.tikz} = \begin{tikzpicture}[scale=1.3, yshift=1.5cm]
    \begin{pgfonlayer}{nodelayer}
        \node [style=none] (0) at (3.75, -2.25) {};
        \node [style=none] (2) at (3.75, -0.75) {};
    \end{pgfonlayer}
    \begin{pgfonlayer}{edgelayer}
        \draw (0.center) to (2.center);
    \end{pgfonlayer}
\end{tikzpicture}
	\end{equation*}

	\begin{equation*}
		\input{frob/7.tikz} = \input{frob/8.tikz}, \hspace*{1cm} \input{frob/9.tikz} = \input{frob/10.tikz}, \hspace*{1cm} \input{frob/11.tikz} = \begin{tikzpicture}[scale=1.3, yshift=1.7cm]
    \begin{pgfonlayer}{nodelayer}
        \node [style=none] (0) at (3.75, -2.25) {};
        \node [style=none] (2) at (3.75, -0.75) {};
    \end{pgfonlayer}
    \begin{pgfonlayer}{edgelayer}
        \draw (0.center) to (2.center);
    \end{pgfonlayer}
\end{tikzpicture}
	\end{equation*}

	\begin{equation*}
		\input{frob/13.tikz} = \begin{tikzpicture}[scale=1.3, yshift=0cm]
    \begin{pgfonlayer}{nodelayer}
        \node [style=scalar] (1) at (5.5, -0.25) {$a+b$};
        \node [style=none] (4) at (5.5, 1) {};
        \node [style=none] (5) at (5.5, -1.5) {};
    \end{pgfonlayer}
    \begin{pgfonlayer}{edgelayer}
        \draw (5.center) to (1);
        \draw (1) to (4.center);
    \end{pgfonlayer}
\end{tikzpicture}, \hspace*{1cm} \begin{tikzpicture}[scale=1.3, yshift=0cm]
    \begin{pgfonlayer}{nodelayer}
        \node [style=scalar] (1) at (5.5, -0.25) {$0$};
        \node [style=none] (4) at (5.5, 1) {};
        \node [style=none] (5) at (5.5, -1.5) {};
    \end{pgfonlayer}
    \begin{pgfonlayer}{edgelayer}
        \draw (5.center) to (1);
        \draw (1) to (4.center);
    \end{pgfonlayer}
\end{tikzpicture} = \input{frob/16.tikz}, \hspace*{1cm} \input{frob/17.tikz} = \input{frob/18.tikz}
	\end{equation*}

	\begin{equation*}
		\begin{tikzpicture}[scale=1.3, yshift=1cm]
		\begin{pgfonlayer}{nodelayer}
			\node [style=Z] (1) at (1, -1.25) {};
			\node [style=X] (3) at (1, -0.75) {};
		\end{pgfonlayer}
		\begin{pgfonlayer}{edgelayer}
			\draw (1) to (3);
		\end{pgfonlayer}
	\end{tikzpicture} = \input{frob/20.tikz}, \hspace*{1cm} \input{frob/21.tikz} = \begin{tikzpicture}[scale=1.3, yshift=0cm]
		\begin{pgfonlayer}{nodelayer}
			\node [style=none] (7) at (6.5, 1) {};
			\node [style=none] (8) at (6.5, -1.5) {};
			\node [style=scalar, fill=white] (9) at (6.5, -0.25) {$ab$};
		\end{pgfonlayer}
		\begin{pgfonlayer}{edgelayer}
			\draw (8.center) to (7.center);
		\end{pgfonlayer}
	\end{tikzpicture}, \hspace*{1cm} \begin{tikzpicture}[scale=1.3, yshift=0cm]
		\begin{pgfonlayer}{nodelayer}
			\node [style=none] (7) at (6.5, 1) {};
			\node [style=none] (8) at (6.5, -1.5) {};
			\node [style=scalar, fill=white] (9) at (6.5, -0.25) {$1$};
		\end{pgfonlayer}
		\begin{pgfonlayer}{edgelayer}
			\draw (8.center) to (7.center);
		\end{pgfonlayer}
	\end{tikzpicture} = \begin{tikzpicture}[scale=1.3, yshift=0cm]
		\begin{pgfonlayer}{nodelayer}
			\node [style=none] (7) at (6.5, 1) {};
			\node [style=none] (8) at (6.5, -1.5) {};
		\end{pgfonlayer}
		\begin{pgfonlayer}{edgelayer}
			\draw (8.center) to (7.center);
		\end{pgfonlayer}
	\end{tikzpicture}
	\end{equation*}

	\begin{equation*}
		\input{frob/25.tikz} = \input{frob/26.tikz}, \hspace*{1cm} \begin{tikzpicture}[scale=1.3, yshift=0.5cm]
		\begin{pgfonlayer}{nodelayer}
			\node [style=Z] (14) at (10, 0) {};
			\node [style=scalar] (17) at (10, -0.75) {$a$};
			\node [style=none] (18) at (10, -1.5) {};
		\end{pgfonlayer}
		\begin{pgfonlayer}{edgelayer}
			\draw (18.center) to (17);
			\draw (17) to (14);
		\end{pgfonlayer}
	\end{tikzpicture} = \begin{tikzpicture}[scale=1.3, yshift=0.5cm]
		\begin{pgfonlayer}{nodelayer}
			\node [style=Z] (19) at (11.25, -0.5) {};
			\node [style=none] (22) at (11.25, -1.25) {};
		\end{pgfonlayer}
		\begin{pgfonlayer}{edgelayer}
			\draw (22.center) to (19);
		\end{pgfonlayer}
	\end{tikzpicture}, \hspace*{1cm} \input{frob/29.tikz} = \input{frob/30.tikz}
	\end{equation*}

	\begin{equation*}
		\begin{tikzpicture}[scale=1.3, yshift=0cm]
		\begin{pgfonlayer}{nodelayer}
			\node [style=X] (36) at (16.75, -1) {};
			\node [style=scalar] (37) at (16.75, -0.25) {$a$};
			\node [style=none] (38) at (16.75, 0.5) {};
		\end{pgfonlayer}
		\begin{pgfonlayer}{edgelayer}
			\draw (38.center) to (37);
			\draw (37) to (36);
		\end{pgfonlayer}
	\end{tikzpicture} = \begin{tikzpicture}[scale=1.3, yshift=0cm]
		\begin{pgfonlayer}{nodelayer}
			\node [style=X] (39) at (18, -0.5) {};
			\node [style=none] (40) at (18, 0.25) {};
		\end{pgfonlayer}
		\begin{pgfonlayer}{edgelayer}
			\draw (40.center) to (39);
		\end{pgfonlayer}
	\end{tikzpicture}, \hspace*{1cm} \input{frob/33.tikz} = \input{frob/34.tikz},  \hspace*{1cm} \input{frob/35.tikz} = \input{frob/36.tikz}
	\end{equation*}

	\begin{equation*}
		\begin{tikzpicture}[scale=1.3, yshift=1cm]
	\begin{pgfonlayer}{nodelayer}
		\node [style=Z] (0) at (3.75, -1) {};
	\end{pgfonlayer}
\end{tikzpicture} = \begin{tikzpicture}[scale=1.3, yshift=1cm]
	\begin{pgfonlayer}{nodelayer}
		\node [style=X] (0) at (3.75, -1) {};
	\end{pgfonlayer}
\end{tikzpicture} = \input{frob/39.tikz}, \hspace*{1cm} \input{frob/40.tikz} = \input{frob/41.tikz} = \begin{tikzpicture}[scale=1.3, yshift=0cm]
	\begin{pgfonlayer}{nodelayer}
		\node [style=none] (3) at (17, 1.5) {};
		\node [style=none] (4) at (17, -0.75) {};
	\end{pgfonlayer}
	\begin{pgfonlayer}{edgelayer}
		\draw (4.center) to (3.center);
	\end{pgfonlayer}
\end{tikzpicture}
	\end{equation*}
\end{lemma}

\begin{definition}
	The $\catname{PROP}$ of affine relations over $\mathbb{F}_2$, $\catname{AffRel}_{\mathbb{F}_2}$, is constructed in the same way as $\catname{LinRel}_{\mathbb{F}_2}$, except a map $ n \to m $ is instead a (possibly empty) affine subspace of $ {\mathbb{F}_2}^n \oplus {\mathbb{F}_2}^m $.
\end{definition}

\begin{lemma}
	$\catname{AffRel}_{\mathbb{F}_2}$ is generated by adding an extra generator to $\catname{LinRel}_{\mathbb{F}_2}$ modulo the following equations:

	\begin{equation*}
		\begin{tikzpicture}[scale=1.3, yshift=1cm]
	\begin{pgfonlayer}{nodelayer}
		\node [style=X] (0) at (3.25, -1.5) {$1$};
		\node [style=none] (1) at (4, -0.25) {};
		\node [style=none] (2) at (4, -2) {};
		\node [style=X] (3) at (3.25, -0.75) {};
	\end{pgfonlayer}
	\begin{pgfonlayer}{edgelayer}
		\draw (0) to (3);
		\draw (2.center) to (1.center);
	\end{pgfonlayer}
\end{tikzpicture} = \input{frob/44.tikz}, \hspace*{1cm} \input{frob/45.tikz} = \begin{tikzpicture}[scale=1.3, yshift=0cm]
	\begin{pgfonlayer}{nodelayer}
		\node [style=X] (0) at (0, 0) {$1$};
		\node [style=none] (2) at (0, 1) {};
		\node [style=none] (3) at (0.5, 1) {};
		\node [style=X] (4) at (0.5, 0) {$1$};
	\end{pgfonlayer}
	\begin{pgfonlayer}{edgelayer}
		\draw (0) to (2.center);
		\draw (4) to (3.center);
	\end{pgfonlayer}
\end{tikzpicture},  \hspace*{1cm} \begin{tikzpicture}[scale=1.3, yshift=0cm]
	\begin{pgfonlayer}{nodelayer}
		\node [style=X] (0) at (0.75, -0.5) {$1$};
		\node [style=Z] (1) at (0.75, 0.5) {};
	\end{pgfonlayer}
	\begin{pgfonlayer}{edgelayer}
		\draw (0) to (1);
	\end{pgfonlayer}
\end{tikzpicture} = \input{frob/48.tikz}.
	\end{equation*}

\end{lemma}

\begin{lemma}
	The functor $(\_)^\perp:\catname{AffRel}_{\mathbb{F}_2} \to \catname{AffRel}_{\mathbb{F}_2}$ defined by:
	\begin{equation*}
	    \input{frob/ih-1.tikz}
		\mapsto
		\input{frob/ih-2.tikz},
		\quad
		\input{frob/ih-3.tikz}
		\mapsto
		\input{frob/ih-4.tikz},
		\quad
		\begin{tikzpicture}[scale=1.3, yshift=-3.3cm]
	\begin{pgfonlayer}{nodelayer}
		\node [style=none] (2) at (7.75, 3.75) {};
		\node [style=none] (3) at (7.75, 2.75) {};
		\node [style=scalar] (5) at (7.75, 3.25) {$a$};
	\end{pgfonlayer}
	\begin{pgfonlayer}{edgelayer}
		\draw (5) to (3.center);
		\draw (5) to (2.center);
	\end{pgfonlayer}
\end{tikzpicture}
		\mapsto
		\begin{tikzpicture}[scale=1.3, yshift=-3.3cm]
	\begin{pgfonlayer}{nodelayer}
		\node [style=none] (2) at (7.75, 3.75) {};
		\node [style=none] (3) at (7.75, 2.75) {};
		\node [style=scalarop] (5) at (7.75, 3.25) {$a$};
	\end{pgfonlayer}
	\begin{pgfonlayer}{edgelayer}
		\draw (5) to (3.center);
		\draw (5) to (2.center);
	\end{pgfonlayer}
\end{tikzpicture},
		\quad
		\begin{tikzpicture}[scale=1.3, yshift=-3.3cm]
	\begin{pgfonlayer}{nodelayer}
		\node [style=none] (2) at (7.75, 3.75) {};
		\node [style=none] (3) at (7.75, 2.75) {};
		\node [style=scalarop] (5) at (7.75, 3.25) {$a$};
	\end{pgfonlayer}
	\begin{pgfonlayer}{edgelayer}
		\draw (5) to (3.center);
		\draw (5) to (2.center);
	\end{pgfonlayer}
\end{tikzpicture}
		\mapsto
		\begin{tikzpicture}[scale=1.3, yshift=-3.3cm]
	\begin{pgfonlayer}{nodelayer}
		\node [style=none] (2) at (7.75, 3.75) {};
		\node [style=none] (3) at (7.75, 2.75) {};
		\node [style=scalar] (5) at (7.75, 3.25) {$a$};
	\end{pgfonlayer}
	\begin{pgfonlayer}{edgelayer}
		\draw (5) to (3.center);
		\draw (5) to (2.center);
	\end{pgfonlayer}
\end{tikzpicture},
	\end{equation*}

	we also have the orthogonal complement:
	$$
		V \mapsto V^\perp := \{  v \in V: \forall w \in V , \langle v,w\rangle = 0\}.
	$$

\end{lemma}

Affine relations are represented as affine subspaces that map between vector spaces, with objects given by natural numbers (indicating vector space dimensions) and morphisms as affine subspaces of ${\mathbb{F}_2}^n \oplus {\mathbb{F}_2}^m$. These subspaces define possible pairings between vectors in the input and output spaces, allowing for transformations with both linear and constant shift components.

In the category $\catname{AffRel}_{\mathbb{F}_2}$, composition is relational: given two affine relations, one from an $n$-dimensional space to an $m$-dimensional space and another from $m$ to an $\ell$-dimensional space, their composition is the set of pairs in ${\mathbb{F}_2}^n \oplus {\mathbb{F}_2}^\ell$ that can be linked by an intermediate vector in ${\mathbb{F}_2}^m$.

For example, a relation from ${\mathbb{F}_2}^2$ to $\mathbb{F}_2$ might be represented by the affine subspace $\{(x, y, z) \in {\mathbb{F}_2}^2 \oplus {\mathbb{F}_2} : z = x + y + 1\}$, where the output is shifted by 1. Another relation from ${\mathbb{F}_2}^3$ to ${\mathbb{F}_2}^2$ could be given by $\{(x, y, z, w, v) \in {\mathbb{F}_2}^3 \oplus {\mathbb{F}_2}^2 : w = x + y, v = y + z\}$, with outputs as specific linear combinations of inputs.

\begin{lemma}
	There is a faithful, strong symmetric monoidal functor $L:\catname{LinRel}_{\mathbb{F}_2} \to \catname{LagRel}_{\mathbb{F}_2}$ given by the following action on the generators of $\catname{LinRel}_{\mathbb{F}_2}$ (i.e.\ doubling, and then changing the colours of one of the copies):
	$$
		\begin{tikzpicture}[scale=2, yshift=1cm]
	\begin{pgfonlayer}{nodelayer}
		\node [style=map] (0) at (-3, -1) {$V$};
		\node [style=none] (1) at (-3, -0.25) {};
		\node [style=none] (2) at (-3, -1.75) {};
	\end{pgfonlayer}
	\begin{pgfonlayer}{edgelayer}
		\draw (1.center) to (0);
		\draw (0) to (2.center);
	\end{pgfonlayer}
\end{tikzpicture}
		\mapsto
		\input{frob/lagrel-2.tikz}
	$$
\end{lemma}

\begin{lemma}
	The forgetful functor $E:\catname{LagRel}_{\mathbb{F}_2} \to \catname{LinRel}_{\mathbb{F}_2}$  is faithful, strong symmetric monoidal.
\end{lemma}

\begin{definition}
	Let $\catname{AffLagRel}_{\mathbb{F}_2}$ denote the monoidal category whose objects are symplectic vector spaces with morphisms are generated by the image of $\catname{LagRel}_{\mathbb{F}_2} \xrightarrow{E} \catname{LinRel}_{\mathbb{F}_2} \to \catname{AffRel}_{\mathbb{F}_2}$ as well as all affine shifts and where the tensor product is the direct sum.
\end{definition}

\subsection{Homology within surface codes and Lagrangian relations}

As discussed in the introduction, braiding defects (or, more generally, performing code deformations) in a surface code induces logical maps on the code space. From the standpoint of error correction, understanding these maps amounts to tracking how the logical operators are transformed under these braiding procedures. We assume the topology of the surface code to be such that we can define logical operators which come in two main flavors:
\begin{itemize}
	\item Loops encircling one or more defects/boundaries, and
	\item Strings connecting pairs of defects or boundaries.
\end{itemize}

Other operators can exist (and most topologies of a surface code patch we will produce other operators), but we only focus on the two types in this work. From a homological perspective, the logical operators are the first homology (or cohomology for the dual lattice) groups on the underlying lattice seen as a chain complex \cite{raussendorf2006fault}. Tracking the evolution of these groups is akin to specifying the Pauli webs or the spacetime stabilisers of maps. The operators and their action (and the first homology groups) at the start of the computation are equivalent to the operators at the end of the computation. From homological algebra, we know that:

\begin{equation}
	H^{1}\!\bigl(C_{\bullet};\,\mathbb{F}_{2}\bigr)
	\quad\text{and}\quad
	H_{2}\!\bigl(C^{\bullet};\,\mathbb{F}_{2}\bigr)
\end{equation}

\noindent
are vector spaces over $\mathbb{F}_2$ where each basis element corresponds to a logical operator. This means that as we evolve the defects, we define a linear relation between the two $\mathbb{F}_2$-vector spaces of logical operators. For surface codes, this space is actually a direct product of vector spaces for primal and dual logical operators. In homological algebra, this becomes a direct product of spaces corresponding to the first homology and the second cohomology groups.

Assume there is an isomorphism between the first homology group and the second cohomology group by Lefschetz duality, that is, that our chain complex comes from an appropriate three-dimensional compact manifold\footnote{For more information on the Lefschetz and Poincar\'e duality see Section 3.3 of Ref.~\cite{hatcher2002algebraic}.}:
\begin{equation}
	H^1(C^{\bullet}, \partial C^{\bullet}; \mathbb{F}_2) \cong H_{2}(C_{\bullet}; \mathbb{F}_2)
\end{equation}
The intersection number between the operators supported on a defect is well-defined, because our assumptions guarantee a primal and dual operator supported on the same defect,
\begin{equation}
	H^1(C^{\bullet}; \mathbb{F}_2) \times H_2(C_{\bullet}; \mathbb{F}_2) \rightarrow \mathbb{F}_2,
\end{equation}
which evaluates to $1$ if the two operators are supported on the same defect and $0$ otherwise. This also captures the commutation relations, since $X$ and $Z$ operators anticommute if when supported on the same defect.

We treat our (co-)homology groups as vector spaces over $\mathbb{F}_2$. Denote the ambient vector space by
\begin{equation}
	V = H_1 \oplus H_1^*.
\end{equation}
Define the pairing as
\[
\langle\,\cdot,\cdot\,\rangle\colon H \times H^* \to \mathbb{F}_2\text,
\]
where $H$ is the vector space of logical operators and $H^*$ is the space of dual operators, and $\langle x,\beta\rangle$ is defined for vectors $x \in H$ and $\beta \in H^*$. In our case, it is given by $x^T \beta$. In other words, we compare the parities of dual and primal operators. This captures the commutation relations of operators defined on a surface code. This allows us to define a symplectic form on \(V\) by:

\begin{equation}
\omega_1\big((x,\alpha),(y,\beta)\big) = \langle x,\beta\rangle + \langle y,\alpha\rangle,
\end{equation}

\noindent
which is zero if the parities across two operators $(x,\alpha),(y,\beta)$ for dual and primal and primal components have the same parity.

Let us now consider the situation where we attempt to map between operators at the beginning of some surface code procedure to the operators at the end. This is a linear relation between the spaces of logical operators on both ``time-slices''. For the second (final) time-slice, define:
\begin{equation}
V_2 = H_2 \oplus H_2^*.
\end{equation}

\noindent
The same construction (via Lefschetz duality and the Kronecker pairing) gives a symplectic form on \(V_2\):
\begin{equation}
\omega_2\big((x,\alpha),(y,\beta)\big) = \langle x,\beta\rangle + \langle y,\alpha\rangle.
\end{equation}
A linear relation (a linear subspace) $g\colon V_1 \rightarrow V_2$ can be defined component-wise:
\begin{equation}
g = g_p \oplus g_d : H_1 \oplus H_1^* \to H_2 \oplus H_2^*,
\end{equation}
where
\[
g_p\colon H_1 \to H_1, \quad g_d\colon H_1^* \to H_2^*.
\]
Let $\Gamma(g)$ be the graph of the relation $g$, which is compatible with the duality pairings; in other words, it is equipped with a symplectic form
\begin{equation}
\Omega = \omega_1 \oplus (-\omega_2),
\end{equation}
and we assume
\begin{equation}
\omega_{1}(v_{1}, w_{1}) \;=\;\omega_{2}(v_{2}, w_{2})
\quad
\text{whenever}
\quad
(v_{1}, v_{2}), \,(w_{1}, w_{2}) \,\in\, \Gamma(g).
\end{equation}

We want to find maximal such graphs (each logical operator on the initial state should be mapped to another logical operator on the final state if it possible). These assumptions and definitions can be interpreted within the context of logical operators in surface codes as enforcing of preservation of commutation relation as the computation progresses. Logical degrees of freedom are defined by the commutation relations and we do not want to alter the structure of the code, possibly merging distinct logical states or splitting one logical state into two. This would no longer describe ``the same logical qubits,” so one would not be performing a valid operation on the original code.

Consider $V_1 \oplus \overline{V_2}$, where \(\overline{V_2}\) denotes \(V_2\) equipped with the symplectic form \(-\omega_1\). Since Lefschetz duality and the Kronecker pairing endow both \(V_1\) and \(V_2\) with symplectic forms (namely, \(\omega_1\) and \(\omega_2\), respectively), we have

\begin{equation}
\Omega\big((v_1, v_2), (w_1, w_2)\big) = \omega_1(v_1,w_1) - \omega_2\big(v_2,w_2\big), \quad
\text{for}
\quad
(v_{1}, v_{2}), \,(w_{1}, w_{2}) \,\in\, \Gamma(g).
\end{equation}

\noindent
Thus, for any two elements \((v_1, v_2)\) and \((w_1, w_2)\) in \(\Gamma(g)\) we obtain:
\begin{equation}
	\Omega\big((v_1, v_2), (w_1, w_2)\big) = \omega_1(v_1,w_1) - \omega_1\big(v_1,w_1\big) = 0.
\end{equation}

\noindent
This shows that \(\Gamma(g)\) is isotropic in \(V_1 \oplus \overline{V_2}\).

To prove that \(\Gamma(g)\) is coisotropic, we recall the assumption that we want to find maximal graphs $\Gamma(g)$. This requires the graph to be maximal isotropic. This means that the graph is coisotropic. Thus, the graph of the relation $g$ is a Lagrangian relation.

\begin{theorem}
The evolution of logical operators in a surface code between two points of the computation is a Lagrangian relation.
\end{theorem}

Under our assumptions, braiding defects realise specific affine Lagrangian relations on the surface--code homology.
Each defect supports two independent classes: a loop that encircles the defect and a string that reaches a boundary.
For a primal defect the loop class is dual (an $X$‑type operator) and the string class is primal (a $Z$‑type operator); for a dual defect the roles are reversed.

When only one species of defects is present, the elementary topological move is a junction that splits one defect into several or merges several into one.
A junction copies the loop class—after the move every participating defect inherits the same encircling loop—and it identifies the string classes: strings from different participating defects to the boundary are pairwise homologous and therefore represent a single equivalence class rather than independent operators.
Because the two species act on opposite homology types, a primal‑type junction copies dual classes while identifying primal classes, whereas a dual‑type junction copies primal classes while identifying dual ones.
Neglecting processes in which the two species meet, these operations give two distinct affine Lagrangian maps on the surface--code homology, one for each species of defect:

\begin{equation*}
	\input{frob/junction-1.tikz} \quad \input{frob/junction-2.tikz}.
\end{equation*}

Of course, we cannot compose sequentially two defects of different species and therefore, a map from above corresponding to one type of defect cannot be composed with a map corresponding to the other type of defect (tensor product is allowed though).

We also specifically designate the following (standard \cite{raussendorf2006fault,raussendorf2007topological}) interpretations:

\begin{minipage}{0.45\textwidth}
	\begin{equation*}
		\left.
		\begin{aligned}
			\input{knot-generators-tikzit/13.tikz}
			\\[1ex]
			\input{knot-generators-tikzit/14.tikz}
		\end{aligned}
		\;\right\}
		\,\mapsto
		\input{frob/cnot-1.tikz}
	\end{equation*}
	\end{minipage}
	\begin{minipage}{0.45\textwidth}
	\begin{equation*}
		\left.
		\begin{aligned}
			\input{knot-generators-tikzit/15.tikz}
			\\[1ex]
			\input{knot-generators-tikzit/16.tikz}
		\end{aligned}
		\;\right\}
		\,\mapsto
		\input{frob/cnot-2.tikz}.
	\end{equation*}
\end{minipage}

The resulting set of Lagrangian relations is not a subcategory of $\catname{LagRel}_{\mathbb{F}_2}$, since it is not closed under sequential composition. Sequentially, we only allow spiders of the same type to be composed.

We now have a language to describe how the logical operators evolve in our setting as the computation progresses. This can be understood as an analysis of a hypothetical situation describing what would happen if these operators were applied. If we actually were to apply one of these operators, it would negate or cancel the hypothetical operator considered in the previous sentence. This is what the semantics of the red decorators are. This is actually what the affine Lagrangian relations describe in $\mathbb{F}_2$.

\begin{corollary}
The evolution of logical operators in a surface code between two points of the computation together with the logical (non-detectable) errors is an affine Lagrangian relation.
\end{corollary}

\subsection{Soundness of $\catname{KNOT}$ for Affine Lagrangian relations}
\label{sec:soundness-afflagrel}

Let $\widehat{\catname{KNOT}}$ be the category generated by the generators of $\catname{KNOT}$ without the red circles and bars.

\begin{definition}
	Define a function $f$ that maps the generators of $\widehat{\catname{KNOT}}$ to ZX-diagrams as follows:

	\begin{equation*}
		 \xmapsto{f} \tikz[tikzfig]{ \draw (0,1) -- (0,-1); }
		\qquad
		 \xmapsto{f} \tikz[tikzfig]{ \draw (0,1) -- (0,-1); }
	\end{equation*}

	\begin{equation*}
		\begin{tikzpicture}[scale=0.65]
	\begin{pgfonlayer}{nodelayer}
		\node [style=none] (0) at (1, 0.5) {};
		\node [style=none] (1) at (1, -0.5) {};
	\end{pgfonlayer}
	\begin{pgfonlayer}{edgelayer}
		\draw [style=primal, in=180, out=-180, looseness=2.50] (0.center) to (1.center);
	\end{pgfonlayer}
\end{tikzpicture}
 \xmapsto{f} \begin{tikzpicture}
	\begin{pgfonlayer}{nodelayer}
		\node [style=none] (12) at (-0.75, -0.25) {};
		\node [style=none] (14) at (0.75, -0.25) {};
		\node [style=none] (15) at (0.75, -0.25) {};
		\node [style=none] (16) at (-0.75, -0.25) {};
	\end{pgfonlayer}
	\begin{pgfonlayer}{edgelayer}
		\draw [bend left=90, looseness=1.75] (16.center) to (15.center);
	\end{pgfonlayer}
\end{tikzpicture}

		\qquad
		\begin{tikzpicture}[scale=0.65]
	\begin{pgfonlayer}{nodelayer}
		\node [style=none] (0) at (1, 0.5) {};
		\node [style=none] (1) at (1, -0.5) {};
	\end{pgfonlayer}
	\begin{pgfonlayer}{edgelayer}
		\draw [style=dual, in=180, out=180, looseness=2.50] (0.center) to (1.center);
	\end{pgfonlayer}
\end{tikzpicture}
 \xmapsto{f} 
		\qquad
		\begin{tikzpicture}[scale=0.65]
	\begin{pgfonlayer}{nodelayer}
		\node [style=none] (0) at (-1, 0.5) {};
		\node [style=none] (1) at (-1, -0.5) {};
	\end{pgfonlayer}
	\begin{pgfonlayer}{edgelayer}
		\draw [style=primal, in=0, out=0, looseness=2.50] (0.center) to (1.center);
	\end{pgfonlayer}
\end{tikzpicture}
 \xmapsto{f} \begin{tikzpicture}
	\begin{pgfonlayer}{nodelayer}
		\node [style=none] (12) at (0.75, 0.25) {};
		\node [style=none] (14) at (-0.75, 0.25) {};
		\node [style=none] (15) at (-0.75, 0.25) {};
		\node [style=none] (16) at (0.75, 0.25) {};
	\end{pgfonlayer}
	\begin{pgfonlayer}{edgelayer}
		\draw [bend left=90, looseness=1.75] (16.center) to (15.center);
	\end{pgfonlayer}
\end{tikzpicture}

		\qquad
		\begin{tikzpicture}[scale=0.65]
	\begin{pgfonlayer}{nodelayer}
		\node [style=none] (0) at (-1, 0.5) {};
		\node [style=none] (1) at (-1, -0.5) {};
	\end{pgfonlayer}
	\begin{pgfonlayer}{edgelayer}
		\draw [style=dual, in=0, out=0, looseness=2.50] (0.center) to (1.center);
	\end{pgfonlayer}
\end{tikzpicture}
 \xmapsto{f} 
	\end{equation*}

	\begin{equation*}
		\input{knot-generators-tikzit/7.tikz} \xmapsto{f} \input{frob/swap.tikz}
		\qquad
		\input{knot-generators-tikzit/8.tikz} \xmapsto{f} \input{frob/swap.tikz}
		\qquad
		\input{knot-generators-tikzit/9.tikz} \xmapsto{f} \input{frob/swap.tikz}
	\end{equation*}

	\begin{equation*}
		\input{knot-generators-tikzit/10.tikz} \xmapsto{f} \input{frob/swap.tikz}
		\qquad
		\input{knot-generators-tikzit/11.tikz} \xmapsto{f} \input{frob/swap.tikz}
		\qquad
		\input{knot-generators-tikzit/12.tikz} \xmapsto{f} \input{frob/swap.tikz}
	\end{equation*}

	\begin{equation*}
		 \xmapsto{f} \begin{tikzpicture}[scale=1.3, yshift=0cm]
	\begin{pgfonlayer}{nodelayer}
		\node [style=Z] (5) at (0, -0.5) {};
		\node [style=none] (6) at (0, 0.5) {};
	\end{pgfonlayer}
	\begin{pgfonlayer}{edgelayer}
		\draw (6.center) to (5);
	\end{pgfonlayer}
\end{tikzpicture}

		\qquad
		 \xmapsto{f} \begin{tikzpicture}[scale=1.3, yshift=0cm]
	\begin{pgfonlayer}{nodelayer}
		\node [style=Z] (5) at (0, 0.5) {};
		\node [style=none] (6) at (0, -0.5) {};
	\end{pgfonlayer}
	\begin{pgfonlayer}{edgelayer}
		\draw (6.center) to (5);
	\end{pgfonlayer}
\end{tikzpicture}

		\qquad
		 \xmapsto{f} \begin{tikzpicture}[scale=1.3, yshift=0cm]
	\begin{pgfonlayer}{nodelayer}
		\node [style=X] (3) at (0, -0.5) {};
		\node [style=none] (4) at (0, 0.5) {};
	\end{pgfonlayer}
	\begin{pgfonlayer}{edgelayer}
		\draw (3) to (4.center);
	\end{pgfonlayer}
\end{tikzpicture}

		\qquad
		\xmapsto{f} \begin{tikzpicture}[scale=1.3, yshift=0cm]
	\begin{pgfonlayer}{nodelayer}
		\node [style=X] (3) at (0, 0.5) {};
		\node [style=none] (4) at (0, -0.5) {};
	\end{pgfonlayer}
	\begin{pgfonlayer}{edgelayer}
		\draw (3) to (4.center);
	\end{pgfonlayer}
\end{tikzpicture}

	\end{equation*}

	\begin{equation*}
		\input{knot-generators-tikzit/13.tikz} \xmapsto{f} \input{frob/gen-3.tikz}
		\qquad
		\input{knot-generators-tikzit/14.tikz} \xmapsto{f} \input{frob/gen-3.tikz}
	\end{equation*}

	\begin{equation*}
		\input{knot-generators-tikzit/15.tikz} \xmapsto{f} \input{frob/gen-4.tikz}
		\qquad
		\input{knot-generators-tikzit/16.tikz}\xmapsto{f} \input{frob/gen-4.tikz}
	\end{equation*}

	\begin{equation*}
		\input{knot-generators-tikzit/19.tikz} \xmapsto{f} \input{frob/grey-2-1.tikz}
		\qquad
		\input{knot-generators-tikzit/21.tikz} \xmapsto{f} \input{frob/white-2-1.tikz}
	\end{equation*}

\end{definition}

\begin{lemma}
	The map $f$ lifts to a functor $F\colon \widehat{\catname{KNOT}} \to \catname{AffRel}_{\mathbb{F}_2}$. It is defined on objects as the length of the bit string and on morphisms as $f$.
\end{lemma}

\begin{proof}
	The axioms of $\catname{KNOT}$ hold in $\catname{AffRel}_{\mathbb{F}_2}$, as they follow from the equivalent rules in ZX-calculus (see \ref{proof:z-functor-zx}). Let us show two examples.
	\begin{equation*}
		\input{knot-generators-tikzit/K5-L.tikz} \mapsto \input{frob/K6-p.tikz} = \input{frob/K6-p-1.tikz} = \input{frob/K6-p-3.tikz} = \begin{tikzpicture}
	\begin{pgfonlayer}{nodelayer}
		\node [style=none] (0) at (-0.5, -1) {};
		\node [style=none] (6) at (0.5, -1) {};
		\node [style=none] (9) at (0.5, 1) {};
		\node [style=none] (11) at (-0.5, 1) {};
	\end{pgfonlayer}
	\begin{pgfonlayer}{edgelayer}
		\draw (11.center) to (0.center);
		\draw (9.center) to (6.center);
	\end{pgfonlayer}
\end{tikzpicture}

		\mapsfrom
		\input{knot-generators-tikzit/K5-R.tikz}
	\end{equation*}

	\begin{equation*}
		\input{knot-generators-tikzit/K8-L.tikz} \mapsto \input{frob/K7-p.tikz} = \input{frob/K7-p-1.tikz} = \input{frob/K7-p-2.tikz}
		\mapsfrom \input{knot-generators-tikzit/K8-R.tikz}
	\end{equation*}

\noindent
Equation \eqref{eq:loop-D-commute} can be proven in the same way as the translation into ZX (see \ref{proof:z-functor-zx}).
All other rules can be proven with an application of a single rule.
\end{proof}

Composing gives a functor $\widehat{G} = L \circ F: \widehat{\catname{KNOT}} \to \catname{LagRel}_{\mathbb{F}_2}$. Finally, we extend $\widehat{G}$ to a functor $G \colon \catname{KNOT} \to \catname{AffLagRel}_{\mathbb{F}_2}$ by mapping the red circles and bars as follows:
\begin{equation*}
	 \mapsto \input{frob/shift-id.tikz}, \hspace*{1cm} \input{knot-generators-tikzit/28.tikz} \mapsto \input{frob/shift-id.tikz}, \hspace*{1cm} \quad \input{knot-generators-tikzit/27.tikz} \mapsto \input{frob/id-shift.tikz}, \hspace*{1cm}  \mapsto \input{frob/id-shift.tikz}.
\end{equation*}
This is indeed a well-defined functor, as
\begin{multline*}
	\input{knot-generators-tikzit/K9-L.tikz} \mapsto \input{frob/K8-p.tikz}  = \input{frob/K8-p-1.tikz}  = \input{frob/K8-p-2.tikz}  = \input{frob/K8-p-3.tikz}  = \\ = \input{frob/K8-p-4.tikz}  = \input{frob/K8-p-5.tikz}  = \input{frob/K8-p-6.tikz}  = \input{frob/K8-p-7.tikz}  = \input{frob/K8-p-8.tikz}
	\mapsfrom \input{knot-generators-tikzit/K9-R.tikz}
\end{multline*}
and the other decorator rules hold trivially.

\begin{corollary}
	$\catname{KNOT}$ is sound for $\catname{AffLagRel}_{\mathbb{F}_2}$.
\end{corollary}

The objects and diagrams in the image of $G$ do not constitute a category: the morphisms in the image are not closed under composition. This lemma, however, allows us to think of the $\catname{KNOT}$ diagrams as pairs of linear relations over $\mathbb{F}_2$.

\begin{corollary}
	\label{soundness-homology}
	$\catname{KNOT}$ is sound for the transformations between the first homology groups of the input and output slices of the surface code computation.
\end{corollary}

\begin{corollary}
	\label{soundness-maps}
	$\catname{KNOT}$ is sound for the logical maps implemented by the braided surface code error correction procedure up to classical byproducts.
\end{corollary}

\section{Lemmas}
\label{appendix-lemmas}
\begin{lemma}
	\label{lemma:ring-around}
	\begin{equation}
		\input{lemmas/lemma1-l.tikz} = \input{lemmas/lemma-2-circs-5.tikz}
	\end{equation}
\end{lemma}
\begin{proof}
	{\allowdisplaybreaks
		\begin{align*}
			\input{lemmas/lemma1-l.tikz} &
			\overset{\substack{\eqref{eq:K-spider} \\ \eqref{eq:spider-disappears}}}{=}
			\input{lemmas/lemma1-1.tikz}
			\overset{\eqref{eq:loop-around-one-wire}}{=} \\[10pt]
			\input{lemmas/lemma1-2.tikz} &
			\overset{\substack{\eqref{eq:loop-D-commute} \\ \eqref{eq:K-spider}}}{=}
			\input{lemmas/lemma1-3.tikz}
			\overset{\substack{\eqref{eq:K-spider} \\ \eqref{eq:loop-D-commute}}}{=}  \\[10pt]
			\input{lemmas/lemma1-5.tikz} &
			\overset{\substack{\eqref{eq:loop-around-one-wire} \\ \eqref{eq:K-spider}}}{=}
			\input{lemmas/lemma-2-circs-5.tikz}
		\end{align*}
	}
\end{proof}

\begin{lemma}
	\label{lemma:two-rings-around}
	\begin{equation}
		\input{lemmas/lemma-2-circs-1.tikz} = \input{lemmas/lemma-2-circs-5.tikz}
	\end{equation}
\end{lemma}
\begin{proof}
	{\allowdisplaybreaks
		\begin{align*}
			\input{lemmas/lemma-2-circs-1.tikz} &
			\overset{\substack{\eqref{lemma:ring-around} \\ \eqref{lemma:ring-around}}}{=}
			\input{lemmas/lemma-2-circs-2.tikz}
			\overset{\eqref{lemma:ring-around}}{=} \\[10pt]
			\input{lemmas/lemma-2-circs-3.tikz} &
			\overset{\substack{\eqref{eq:loop-D-commute} \\ \eqref{eq:K5}}}{=}
			\input{lemmas/lemma-2-circs-4.tikz}
			\overset{\eqref{lemma:ring-around}}{=}
			\input{lemmas/lemma-2-circs-5.tikz}
		\end{align*}
	}
\end{proof}

\begin{lemma}
	\label{lemma-duality}
	\begin{equation}
		\input{lemmas/lemma2-l.tikz} = \input{lemmas/lemma2-r-b.tikz}
	\end{equation}
\end{lemma}
\begin{proof}
	{\allowdisplaybreaks
		\begin{align*}
			\input{lemmas/lemma2-l.tikz} &
			\overset{\substack{\eqref{eq:loop-around-one-wire} \\ \eqref{eq:K5}}}{=}                    \\[10pt]
			\input{lemmas/lemma2-1.tikz} &
			\overset{\eqref{eq:loop-D-commute}}{=}                             \\[10pt]
			\input{lemmas/lemma2-2.tikz} &
			\overset{\eqref{eq:K-spider}}{=}                             \\[10pt]
			\input{lemmas/lemma2-4.tikz} &
			\overset{\eqref{eq:loop-around-one-wire}}{=}
			\input{lemmas/lemma2-r-b.tikz}
		\end{align*}
	}
\end{proof}

\begin{lemma}
	\label{lemma:complex-rule}
	\input{knot-generators-tikzit/K13-L.tikz} = \input{knot-generators-tikzit/K13-R.tikz}
\end{lemma}

\begin{proof}
	{\allowdisplaybreaks
		\begin{align*}
			\input{lemmas/lemma-4-l.tikz} &
			\overset{\eqref{lemma:ring-around}}{=}
			\input{lemmas/lemma-4-1.tikz}
			\overset{\substack{\eqref{lemma-duality} \\ \eqref{eq:loop-D-commute} \\ \eqref{eq:K5}}}{=} \\[10pt]
			\input{lemmas/lemma-4-2.tikz} &
			\overset{\substack{\eqref{eq:loop-D-commute}  \\ \eqref{eq:K5}}}{=}
			\input{lemmas/lemma-4-3.tikz}
			= \\[10pt]
			\input{lemmas/lemma-4-r.tikz}
		\end{align*}
	}
\end{proof}

\begin{lemma}\label{lemma:green-spider-zx-knot}
	\begin{equation}
		\input{zx-knot/spider1.tikz} = \input{zx-knot/spider1-r.tikz} \label{green-spider-zx-knot} \end{equation}
\end{lemma}

\begin{proof}
	{\allowdisplaybreaks
	\begin{align*}


		\end{align*}
	}
\end{proof}

\begin{lemma} \label{lemma:red-spider-copy-thr-green-zx-knot}
	\begin{equation}
		\input{zx-knot/spider4.tikz} = \input{zx-knot/spider4-r.tikz} \label{red-spider-copy-thr-green-zx-knot} \end{equation}
\end{lemma}

\begin{proof}
	Follows from eq. \eqref{eq:loop-around-one-wire}, eq. \eqref{eq:overcross-kill}, and eq. \eqref{eq:K5}  similarly to \ref{lemma:red-spider-zx-knot-state}.
\end{proof}

\begin{lemma} \label{lemma:green-spider-copy-thr-red-zx-knot}
	\begin{equation}
	\input{zx-knot/spider5.tikz} = \input{zx-knot/spider5-r.tikz} \label{green-spider-copy-thr-red-zx-knot} \end{equation}
\end{lemma}

\begin{proof}
	Follows from eq. \eqref{eq:loop-around-one-wire}, eq. \eqref{eq:overcross-kill}, and eq. \eqref{eq:K5}  similarly to \ref{lemma:red-spider-zx-knot-state}.
\end{proof}

\begin{lemma} \label{lemma:green-spider-zx-knot-state}
	\begin{equation}
		\input{zx-knot/spider6.tikz} = \input{zx-knot/spider6-r.tikz} \label{green-spider-zx-knot-state} \end{equation}
\end{lemma}

\begin{proof}
	Follows from eq. \eqref{eq:loop-around-one-wire}, eq. \eqref{eq:overcross-kill}, and eq. \eqref{eq:K5}  similarly to \ref{lemma:red-spider-zx-knot-state}.
\end{proof}

\begin{lemma} \label{lemma:bialgebra}
	\begin{equation}
		\label{bialgebra}
		\input{zx-knot/bialgebra-l.tikz} = \input{zx-knot/bialgebra-r.tikz}
	\end{equation}
\end{lemma}

\begin{proof}
	{\allowdisplaybreaks
	\begin{align*}
		\input{zx-knot/proofs/bialgebra/1.tikz} &
		  \overset{\eqref{lemma:ring-around}}{=}
		\input{zx-knot/proofs/bialgebra/2.tikz}
		  \overset{\eqref{lemma:two-rings-around}}{=} \\[10pt]
		\input{zx-knot/proofs/bialgebra/3.tikz}  &
		  \overset{\eqref{lemma:ring-around}}{=}
		\input{zx-knot/proofs/bialgebra/4-1.tikz}
		  \overset{\substack{\eqref{eq:loop-D-commute} \\ \eqref{lemma:two-rings-around}}}{=}                                        \\[10pt]
		\input{zx-knot/proofs/bialgebra/5-1.tikz}  &
		  \overset{\eqref{lemma:ring-around}}{=}
		\input{zx-knot/proofs/bialgebra/7-1.tikz}
		  \overset{\eqref{eq:loop-D-commute}}{=}                                  \\[10pt]
		\input{zx-knot/proofs/bialgebra/8-1.tikz}  &
		  \overset{\eqref{eq:loop-D-commute}}{=}
		\input{zx-knot/proofs/bialgebra/9-1.tikz}
		  \overset{\eqref{eq:loop-D-commute}}{=}                                        \\[10pt]
		\input{zx-knot/proofs/bialgebra/10-1.tikz}  &
		  \overset{\eqref{eq:K5}}{=}
		\input{zx-knot/proofs/bialgebra/11-1.tikz}
		  =                                     \\[10pt]
		\input{zx-knot/proofs/bialgebra/12-1.tikz}
		&=
		\input{zx-knot/proofs/bialgebra/13-1.tikz}
		  \overset{\eqref{eq:loop-D-commute}}{=}                                             \\[10pt]
		\input{zx-knot/proofs/bialgebra/14-1.tikz}  &
		  \overset{\eqref{lemma:complex-rule}}{=}
		\input{zx-knot/proofs/bialgebra/15-1.tikz}
		  \overset{\eqref{lemma:ring-around}}{=}                                                  \\[10pt]
		\input{zx-knot/bialgebra-r.tikz}
	\end{align*}
	}
\end{proof}

\begin{lemma} \label{lemma:bar-copies}
	\begin{equation}
		\label{bar-copies}
		\input{zx-knot/spider7.tikz} = \input{zx-knot/spider7-r.tikz}
	\end{equation}
\end{lemma}

\begin{proof}
	Holds by applying rules \eqref{eq:K9}, \eqref{eq:K10}, \eqref{eq:K-bar-slide}, \eqref{eq:K-circle-slide}, \eqref{eq:K7}
\end{proof}

\begin{lemma} \label{lemma:circle-passes}
	\begin{multline}
		\label{circle-passes}
		\input{zx-knot/spider8.tikz} = \input{zx-knot/spider8-m.tikz} = \input{zx-knot/spider8-r.tikz}
	\end{multline}
\end{lemma}

\begin{proof}
	Holds by applying rules \eqref{eq:K9}, \eqref{eq:K10}, \eqref{eq:K-bar-slide}, \eqref{eq:K-circle-slide}, \eqref{eq:K7}
\end{proof}

\begin{lemma} \label{lemma:bar-passes}
	\begin{multline}
		\label{bar-passes}
		\input{zx-knot/spider9.tikz} = \input{zx-knot/spider9-m.tikz} = \input{zx-knot/spider9-r.tikz}
	\end{multline}
\end{lemma}

\begin{proof}
	Holds by applying rules \eqref{eq:K9}, \eqref{eq:K10}, \eqref{eq:K-bar-slide}, \eqref{eq:K-circle-slide}, \eqref{eq:K7}
\end{proof}

\begin{lemma} \label{lemma:circle-copies}
	\begin{equation}
		\label{circle-copies}
		\input{zx-knot/spider10.tikz} = \input{zx-knot/spider10-r.tikz}
	\end{equation}
\end{lemma}

\begin{proof}
	Holds by applying rules \eqref{eq:K9}, \eqref{eq:K10}, \eqref{eq:K-bar-slide}, \eqref{eq:K-circle-slide}, \eqref{eq:K7}
\end{proof}

\end{document}